\documentclass[journal]{IEEEtran}

\usepackage{setspace}

\usepackage{graphics,
           psfrag,
           epsfig,
           amsthm,
           cite,
           amssymb,
           url,
           dsfont,
           subfigure,
           algorithm,
           algorithmic,
           balance,
           enumerate,
           color,
           setspace
}
\usepackage{amsmath}

\usepackage{epstopdf}

\newtheorem{definition}{Definition}

\newtheorem{lemma}{Lemma}

\newtheorem{theorem}{Theorem}

\newtheorem{proposition}{Proposition}

\newtheorem{remark}{Remark}

\newcommand{\eref}[1]{(\ref{#1})}
\newcommand{\sref}[1]{Section~\ref{#1}}
\newcommand{\appref}[1]{Appendix~\ref{#1}}
\newcommand{\fref}[1]{Fig.~\ref{#1}}

\newcommand{\pref}[1]{Proposition~\ref{#1}}
\newcommand{\cref}[1]{Constraint~\ref{#1}}
\newcommand{\thref}[1]{Theorem~\ref{#1}}
\newcommand{\lref}[1]{Lemma~\ref{#1}}

\newcommand{\rmref}[1]{Remark~\ref{#1}}
\newcommand{\algref}[1]{Algorithm~\ref{#1}}

\newcommand{\ignore}[1]{}

\IEEEoverridecommandlockouts
\begin{document}

\title{\vspace{-0.5cm}Delay Reduction in Multi-Hop Device-to-Device Communication using Network Coding}

\author{
   \authorblockN{Ahmed Douik, \textit{Student Member, IEEE}, Sameh Sorour, \textit{Member, IEEE}, Tareq Y. Al-Naffouri, \textit{Member, IEEE},\\ Hong-Chuan Yang, \textit{Senior Member, IEEE}, and Mohamed-Slim Alouini, \textit{Fellow, IEEE}}
    
    {\thanks {A part of this paper is accepted in IEEE International Symposium on Network Coding (NetCod' 2015), Sydney, Australia. The pre-print version is available at \cite{5456269}.
    
    Ahmed Douik and Mohamed-Slim Alouini are with Computer, Electrical and Mathematical Sciences and Engineering (CEMSE) Division at King Abdullah University of Science and Technology (KAUST), Thuwal, Makkah Province, Saudi Arabia, email: \{ahmed.douik,slim.alouini\}@kaust.edu.sa.

Sameh Sorour is with the Electrical Engineering Department, King Fahd University of Petroleum and Minerals (KFUPM), Dhahran, Eastern Province, Saudi Arabia, email:samehsorour@kfupm.edu.sa.

Tareq Y. Al-Naffouri is with both the CEMSE Division at King Abdullah University of Science and Technology (KAUST), Thuwal, Makkah Province, Saudi Arabia, and the Electrial Engineering Department at King Fahd University of Petroleum and Minerals (KFUPM), Dhahran, Eastern Province, Saudi Arabia, e-mail: tareq.alnaffouri@kaust.edu.sa.

Hong-Chuan Yang is with the Department of Electrical and Computer Engineering, University of Victoria (UVIC), British Columbia, Canada, email: hy@uvic.ca.
}}
\vspace{-0.8cm}    }

\maketitle

\IEEEoverridecommandlockouts

\begin{abstract}
This paper considers the problem of reducing the broadcast decoding delay of wireless networks using instantly decodable network coding (IDNC) based device-to-device (D2D) communications. In a D2D configuration, devices in the network can help hasten the recovery of the lost packets of other devices in their transmission range by sending network coded packets. Unlike previous works that assumed fully connected network, this paper proposes a partially connected configuration in which the decision should be made not only on the packet combinations but also on the set of transmitting devices. First, the different events occurring at each device are identified so as to derive an expression for the probability distribution of the decoding delay. The joint optimization problem over the set of transmitting devices and the packet combinations of each is, then, formulated. The optimal solution of the joint optimization problem is derived using a graph theory approach by introducing the cooperation graph and reformulating the problem as a maximum weight clique problem in which the weight of each vertex is the contribution of the device identified by the vertex. Through extensive simulations, the decoding delay experienced by all devices in the Point to Multi-Point (PMP) configuration, the fully connected D2D (FC-D2D) configuration and the more practical partially connected D2D (PC-D2D) configuration are compared. Numerical results suggest that the PC-D2D outperforms the FC-D2D and provides appreciable gain especially for poorly connected networks.
\end{abstract}

\begin{keywords}
Instantly decodable network coding, device-to-device communications, delay reduction, partially connected network, maximum weight clique problem.
\end{keywords}

\section{Introduction} \label{sec:intro}

\subsection{Background}

\emph{Network Coding (NC)} \cite{850663} is a promising technique to significantly improve the throughput (i.e., the network capacity) and to minimize delay over wireless erasure channels. These benefits are of great interest for the proliferation and spread of real-time applications which require quick and reliable packet transmission over lossy channels with low latency \cite{6283957}, such as streaming, cellular, satellite networks, and Internet television.

Two classes of NC can be distinguished in the literature, namely the Random Network Coding (RNC) \cite{1228459,4015713} and the Opportunistic Network Coding (ONC) \cite{4937121,1159942}. While the sender in RNC combines packets using random coefficients, ONC exploits the diversity of received and lost packets to generate the network encoded packet combinations. Even though RNC is optimal in reducing the number of transmissions even without feedback, it is not suitable for real-time applications since it does not enable progressive decoding of the frame. Moreover, unlike ONC, the computation complexity of RNC is prohibitive for real-time applications.

This paper is interested in delay sensitive broadcast applications in which each device should receive all the packets within a frame with minimum delay \cite{5753573}. A suitable technique for such applications is the Instantly Decodable Network Coding (IDNC) \cite{5425315}. In this subclass of ONC, the sender encodes the packets using binary XOR and devices decode them by the same means, which is an important property that ensures fast encoding/decoding and overcomes the computationally expensive matrix inversion operations. In addition, non-instantly decodable packets are not stored, which eliminates the need for buffers and allows the design of cost-efficient receivers. Thanks to its numerous benefits, IDNC was an subject of intensive research in the past few years \cite{5753573,6381028,4313060,5425315,5683677,6809217,6362857,5191398,164579,20112430,6570827}.

In the wireless medium, packet losses, occurring due to many phenomena (e.g., the mobility and the propagation environment), can be seen as packet erasures at higher communication layers \cite{1208721}. The erasure nature of the links affects the ability of devices to decode the information flow synchronously and thus affects the delivery of meaningful data. As a consequence, a better use of the channel does not usually translates to a better throughput at higher communication layers \cite{1208721} which motivates the definition of the various delay metric in NC. The commonly used delay metrics in IDNC are the completion time and the decoding delay. The former definition considers the delay as the overall transmission time and the latter as the individual delay when a transmitted packet brings no new information at its reception instant. This paper considers the minimization of the decoding delay as it represents a crucial step to study the completion time reduction \cite{13051412}.

\subsection{Related Work}

Determining the packet combinations for the whole recovery phase to optimally reduce the decoding delay is shown to be intractable even for erasure free \cite{6503457} or the off-line \cite{5205612} scenarios (i.e., scenarios in which the erasure events' are known in advance) with only three devices. In order to overcome the complexity, to the best of the authors' knowledge, the most adopted approach \cite{5425315,5683677,6809217,6362857,5191398,164579,20112430,6570827} is to study the on-line decoding delay. Therefore, the decoding delay minimization problem is formulated for each recovery transmission.

In all aforementioned works, the wireless centralized sender of a Point to Multi-Point (PMP) network (such as cellular, Wi-Fi and roadside to vehicle systems) is the only transmitter in charge of both the packet sending and recovery processes. This approach consumes a lot of the sender resources and threats its ability to deliver the packets with the desired rates. The problem is expected to further escalate in the next generation mobile radio system ($5$G), since the required data rates and the Quality of Service (QoS) requirements are becoming even more constraining \cite{6824752,6736746}. The notion of Device-to-device (D2D) communications, introduced in \cite{6404743}, is a suitable technique to overcome the problem. In a D2D configuration, devices can take care of the packet recovery process by exchanging packets with each other over short range and possibly more reliable channels. This D2D model also provides fast and secure data communications over ad-hoc networks, e.g., wireless sensor networks. Minimizing delays in such IDNC-based D2D systems is of great interest.

Unlike the PMP scenario in which the decision is made only on the packet combination to be transmitted, the decision in a D2D environment should be made also on the set of transmitting devices in every transmission to achieve the best network performance. Aboutorab et al. \cite{6620795} investigate the problem of minimizing the sum decoding delay in a centralized D2D environment, where the term centralized refers to the configuration in which a \emph{leader} takes all the decisions. The authors in \cite{6970860} extend the study to the imperfect feedback D2D networks. In a fully distributed D2D system, reference\cite{14122014} considers the completion time minimization using game theory as a tool to improve the distributed solution. The formulation is further extended to the decoding delay reduction in \cite{445913}.

These prior works on IDNC-based D2D systems assumes that the network is fully connected (FC-D2D), i.e., each device can target all other devices over one-hop transmission, and thus only one device can transmit at each transmission slot. This assumption of FC-D2D may not apply in some realistic scenarios due to the short device transmission ranges and their widespread over the large cell area. Longer-range D2D transmissions can easily limit the desirable property of more reliable communications between devices compared to those from the centralized sender. This work proposes to study of the decoding delay minimization in the case of partially connected D2D networks (PC-D2D). The partially connected configuration, in which the fully connected can be seen as a particular case, add a new dimension to the problem as many devices can communicate simultaneously each with different packet combination.

\subsection{Contribution}

The main contribution of this paper is to study the reduction of the decoding delay of IDNC in partially connected D2D networks. Whereas in fully connected D2D networks the problem of selecting the transmitting devices and the problem of selecting the packet combinations to be transmitted can be addressed separately, in partially connected systems the two problems are interdependent. Therefore, the optimization over the set of transmitting devices and their packet combinations should be addressed jointly. Although the paper focuses on a centralized solution, the decoding delay analysis presented in this paper will serve as a reference to future research direction on fully distributed systems since it provides a lower bound on the achievable decoding delay. Moreover, the decoding delay formulation is useful to study the completion time reduction problem \cite{13051412}.

The paper first identify the different expected decoding delay for each device in the network. These expressions are then exploited to formulate the minimum decoding delay problem in IDNC-based PC-D2D networks. To solve the problem, the paper first proposes a decoupling approach. The variables of the problem are shown to be separable in the interference-less scenario. The paper addresses the interference-less decoding delay reduction using a graph theory approach by introducing the local IDNC graph for packet generation and the cooperation graph and reformulating the problem as a maximum weight clique problem that can be efficiently solved using existing literature from graph theory \cite{16513519,13265492}. The solution in the interference-less scenario is further combined with a clustering technique to solve the original problem. Finally, the paper proposes a method to generate the minimum number of clusters to reach the optimal solution of the decoding delay reduction problem. Simulations results show that the proposed solution displays appreciable gain as compared with the fully connected D2D network approach and the PMP configuration.

The rest of this paper is organized as follows: \sref{sec:sys} introduces the system model and parameters. The decoding delay minimization problem in partially connected D2D network is formulated in \sref{sec:prob}. \sref{sec:no} solves the problem in the particular interference-less case. In \sref{sec:ext} the decoding delay minimization problem is solved. Simulation results are illustrated and discussed in \sref{sec:sim} before concluding in \sref{sec:conc}.

\section{System Model and Notations} \label{sec:sys}

\subsection{System Model and Parameters}

Consider a wireless base-station (BS) that desires to transmit $N$ different packets to a set $\mathcal{M}$ of $M$ device. Let $\mathcal{N}$ denote the frame consisting of the various source packets. The BS begins by broadcasting the $N$ source packets sequentially. Each device listens and sends an acknowledgement upon each successful reception. The probability of a packet loss at device $i$ when the BS is transmitting is $q_i,\ \forall \ i \in \mathcal{M}$. At the end of this initial phase, we assume that each packet of the frame $\mathcal{N}$ is acknowledged by at least one device. Otherwise, the BS keeps broadcasting the packet until the condition is verified. After the initialization phase, the packets of the frame can be in one of the following two sets for each device $i$:
\begin{itemize}
\item The Has set $\mathcal{H}_i$: the set of packets received by device $i$.
\item The Wants set $\mathcal{W}_i$: the set of packets lost by device $i$.
\end{itemize}

In the recovery phase, devices cooperate to ensure that everyone successfully receives all the $N$ packets. The transmitting devices and the packet combinations are chosen using the available information including the expected erasure patterns of the network links, and the diversity of received/lost packets across the network. In this phase, the coded packets can be one of the following options for device $i$:
\begin{itemize}
\item \emph{Non-innovative:} A packet is non-innovative for device $i$ if it does not contain \emph{any} packet from $\mathcal{W}_i$.
\item \emph{Instantly Decodable:} A packet is instantly decodable for device $i$ if it contains exactly \emph{one packet} from $\mathcal{W}_i$.
\item \emph{Non-Instantly Decodable:} A packet is non-instantly decodable for device $i$ if it contains more than one packet from $\mathcal{W}_i$.
\end{itemize}

The decoding delay \cite{5683677,14122014} is defined as follows:
\begin{definition}
At any recovery phase transmission, a device $i$, with non-empty Wants sets, experiences one unit increase of decoding delay if he cannot hear exactly one transmission or if he hears a packet that is either non-instantly decodable or non-innovative.
\end{definition}

Let $\mathbf{C}=[c_{ij}],\ \forall (i,j) \in \mathcal{M}^2$ be the connectivity matrix representing the connectivity between any pair of devices defined as follows:
\begin{align}
c_{ij} =
\begin{cases}
1 \hspace{0.1cm}& \text{if devices }i,j \text{ are connected}\\
0 \hspace{0.1cm}& \text{otherwise}
\end{cases}\quad, \forall\ (i,j) \in \mathcal{M}^2 \nonumber.
\end{align}

Note that $\mathbf{C}$ is a symmetric matrix that depends on the network topology (i.e., the relative positions of devices in the network). This paper considers a network with a general topology. However, each device should be able to target any other device through single or multi-hop transmission (via the intermediate nodes). In graph theory terms, the graph representing the devices is \emph{connected} or equivalently the matrix $\mathbf{C}$ is \emph{connected}.

Let $\mathcal{C}_i$ be the coverage zone of device $i$ defined as the set of devices in the transmission range of device $i$ (i.e., $\mathcal{C}_i = \{j \in \mathcal{M} \ | \ c_{ij}=1\}$), and let $\mathbf{P}=[p_{ij}],\ \forall (i,j) \in \mathcal{M}^2$ denote the packet erasure probability from device $i$ to device $j$.

\subsection{Notations}

The notation $\overline{X}$, where $X$ is a set in the ensemble $E$, refers to the complementary of the set, i.e., $\overline{X}=E \setminus X$. The notation $\mathcal{P}(E)$ refers to the power set of the ensemble $E$. In other words, $\mathcal{P}(E)$ is the set of all the subsets of the ensemble $E$. A partition $\left\{X_i\right\}_{1 \leq i \leq n}$ of the ensemble $E$ is denoted by $E = \bigoplus_{i = 1}^n X_i$. Let $|X|$ denote the cardinality of the set $X$.

\section{Minimum Decoding Delay Problem Formulation}\label{sec:prob}

This section first formulates the minimum decoding delay problem in partially connected D2D networks. The problem is expressed as a joint optimization over the set of transmitting devices and the packet combinations to be transmitted. Further, the section illustrates the optimal solution to the packet mix optimization for a fixed set of transmitting devices by introducing the local IDNC graph for packet generation.

\subsection{Problem Formulation}

Let $\mathcal{A} \in \mathcal{P}(\mathcal{M})$ be the set of transmitting devices and let $\mathcal{T}(\mathcal{A})$ be the set of non-transmitting devices in interference. In other words, $\mathcal{T}(\mathcal{A})$ is the set of devices that can hear multiple transmissions from the set of transmitting devices $\mathcal{A}$. The mathematical definition of this set is:
\begin{equation}
\mathcal{T}(\mathcal{A}) = \left\{i\notin \mathcal{A}~ \middle| ~\exists\ (m,n) \in \mathcal{A}^2,\ m \neq n, \ i \in \mathcal{C}_n \cap \mathcal{C}_m\right\}. \nonumber
\end{equation}
Define $\mathcal{S}(\mathcal{A})$ as the set of devices that are not in the transmission range of any transmitting device. Formally, the set is written as:
\begin{align}
\mathcal{S}(\mathcal{A}) = \left\{ i \in \mathcal{M}~ \middle| ~\nexists\ j \in \mathcal{A}, \ i \in \mathcal{C}_j\right\}.
\end{align}

Define the opportunity zone $\mathcal{O}_i(\mathcal{A})= \mathcal{C}_i \setminus \left( \mathcal{A} \cup \mathcal{T}(\mathcal{A}) \right)$ as the set of devices that can be targeted by device $i$ and can decode a packet from the transmission.

Let $M_w$ be the set of devices having non-empty Wants set and let $\kappa_i(\mathcal{A})$ be the packet combination to be transmitted by device $i \in \mathcal{A}$. Define $\tau_i(\kappa_i(\mathcal{A}),\mathcal{A})$ as the set of targeted devices by device $i$ when he is sending the packet combination $\kappa_i$ and all devices in $\mathcal{A}$ are transmitting.
\begin{remark}
The variables defined above should be all a function of the set of transmitting devices $\mathcal{A}$. However, for notation convenience, the set will be dropped unless it is required (e.g., we should write $\mathcal{T}$ and $\tau_i(\kappa_i)$ instead of $\mathcal{T}(\mathcal{A})$ and $\tau_i(\kappa_i(\mathcal{A}),\mathcal{A})$).
\label{rm1}
\end{remark}

The minimum decoding delay problem can be formulated as a joint optimization problem over the set of transmitting devices and their packet combinations as illustrated in the following theorem:
\begin{theorem}
The decoding delay reduction problem in partially connected D2D network can be formulated as follows:
\begin{subequations}
\label{original}
\begin{align}
&\max_{ \mathcal{A} \in \mathcal{P}(\mathcal{M}) } -|\mathcal{A} \cap M_w| - |\mathcal{T} \cap M_w| - |\mathcal{S} \cap M_w| \nonumber \\
& \qquad + \sum_{i \in \mathcal{A}} \left( \sum_{j \in \tau_{i}(\kappa_{i}^*) } 1-p_{ij} \right) \label{eq12} \\
&\text{subject to } \nonumber\\
&\kappa_{i}^*(\mathcal{A}) = \arg \max_{\kappa_{i} \in \mathcal{P}(\mathcal{H}_{i}) } \left( \sum_{j \in \tau_{i}(\kappa_{i}) } 1-p_{ij} \right), \forall \ i \in \mathcal{A}. \label{eq11}
\end{align}
\end{subequations}
\label{th1}
\end{theorem}
\begin{proof}
The proof can be found in \appref{ap7}.
\end{proof}

Let the objective function \eref{eq12} of the optimization problem \eref{original} be called the outer problem, and the constraint \eref{eq11} be called the inner problem. Finding the global optimal solution to the optimization problem \eref{original} may require a search over all the sets of transmitting devices and their possible packet combinations. Such solution is clearly infeasible for any reasonable sized network. In the next subsection, the paper proposes finding the optimal packet mix for a fixed set of transmitting devices, i.e., the optimal $\kappa_{i}^*, \ i \in \mathcal{A}$ for a fixed set $\mathcal{A}$. In other words, the global optimal solution to the inner problem \eref{eq11}.

\subsection{Local IDNC Graph for Packet Generation}

Let the set of transmitting devices be fixed to $\mathcal{A}$. To solve the optimization problem \eref{eq11}, this subsection rely on a graph theoretical model first introduced in the context of PMP networks in \cite{5683677}. Therefore, this section first extends the IDNC graph to account for the restriction in the packet generation and the partial connection of the graph. In this context, let such graph be called the \emph{local IDNC graph}. Finally, the paper reformulates the problem \eref{eq11} as a maximum weight clique problem in the local IDNC graph.

The IDNC graph is a tool introduced for the PMP model to determine both all possible XOR-based coded combinations, and the devices that can instantly decode each of them. In contract with PMP model that permit the generation of all the packet mixes, in D2D environment, each device can generate coding combinations only from the packets it possesses (i.e., packets in his Has set). Further, while the sender in the fully connected network can intend packets to all devices, in partially connected systems, each device can target only the devices in its transmission range. This subsection illustrates how the different devices can build similar local IDNC graphs. Naturally, the local graph depends on the set of transmitting devices. However following the simplification adopted in \rmref{rm1}, the set of transmitting users is dropped.

To construct the local IDNC graph $\mathcal{G}_i(\mathcal{V}_i, \mathcal{E}_i)$ of device $i \in \mathcal{A}$, a vertex $v_{kl} \in \mathcal{V}_i$ is generated for each packet $l \in (\mathcal{W}_k \cap \mathcal{H}_i),~ \forall~ k \in \mathcal{O}_i$. An edge in $\mathcal{E}_i$ connecting each vertices $v_{kl}$ and $v_{mn}$ is created if one of the two following conditions is true:
\begin{itemize}
\item $l=n \Rightarrow$ Packet $l$ is needed by both devices $k$ and $m$.
\item $l \in \mathcal{H}_m$ and $n \in \mathcal{H}_k \Rightarrow$ The packet combination $l\oplus n$ is instantly decodable for both devices $k$ and $m$.
\end{itemize}

The following lemma characterizes the solution of the decoding delay reduction problem \eref{eq11} for a fixed set of transmitting devices $\mathcal{A}$.

\begin{lemma}
The optimal solution $\kappa_i^*$ of the optimization problem \eref{eq11} for device $i \in \mathcal{A}$ is the maximum weight clique in the local IDNC graph $\mathcal{G}_i$ of device $i$ in which the weight of each vertex $v_{kl}$ is $1-p_{ik}$.
\label{l4}
\end{lemma}

\begin{proof}
The proof can be found in \appref{ap9}.
\end{proof}

Given the solution to the decoding delay reduction problem \eref{eq11} for a fixed set of transmitting devices, the global optimal solution to the decoding delay reduction problem \eref{original} can be obtained by solving problem \eref{eq11} for all possible set of transmitting devices. Since an exhaustive search over the set of devices $\mathcal{M}$ is needed, the complexity of such algorithm is in the order of $2^M \mathfrak{f}$, where $\mathfrak{f}$ is the complexity of solving \eref{eq11} (i.e., the complexity of solving the maximum weight clique problem). Therefore, such approach is clearly infeasible.

The rest of the paper shows an efficient method for reaching the global optimal solution of the optimization problem \eref{original}. In the next section, the optimization problem \eref{original} is globally solved when imposing restrictions on the set of transmitting devices that cooperation is allowed only when no device is in interference. In other words, the problem is solved under the constraint $\mathcal{T} = \varnothing$. The section solves the problem by introducing the \emph{cooperation graph} and reformulating the problem as a maximum weight clique problem in that graph. Finally, the paper shows that employing a clustering approach and using both the cooperation graph and the local IDNC graph, the global optimal solution of \eref{original} can be achieved.

\section{Interference-less Cooperative Solution}\label{sec:no}

Due to the interdependence of the set of transmitting devices and transmitted packets, the variables are not separable as shown in \fref{fig:network} which presents a system and its associated feedback matrix.

\begin{figure}[t]
\centering
\includegraphics[width=1\linewidth]{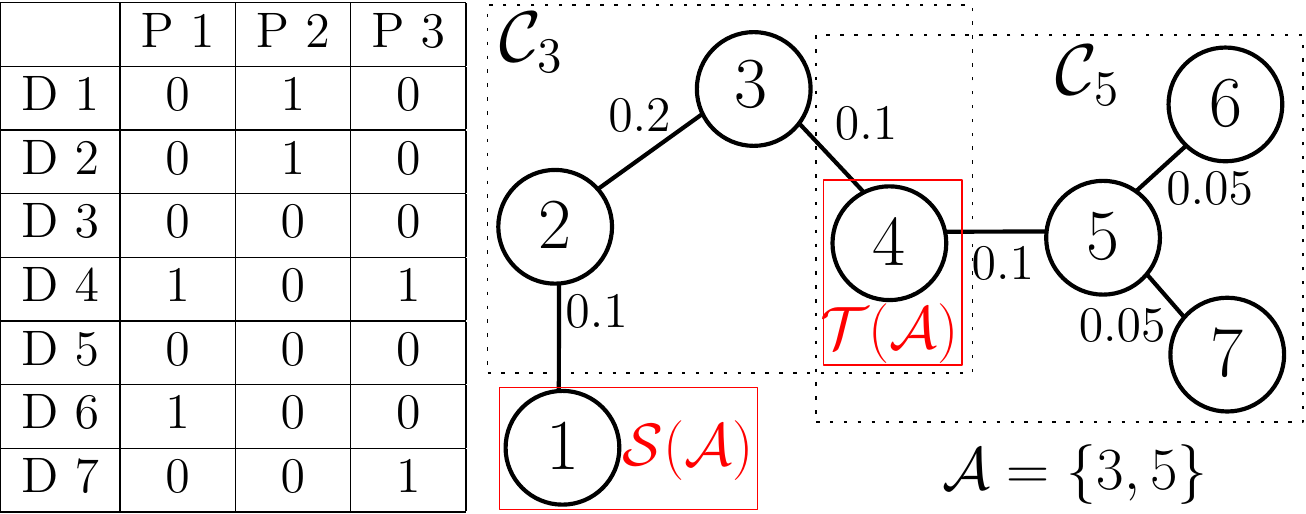}\\
\caption{Network composed of $7$ devices and $3$ packets. The feedback matrix represents the distribution of lost ($1$) and received packets ($0$) at each device. The erasure probabilities between devices is presented on the edges.}
\label{fig:network}
\end{figure}

Given the network configuration and the distribution of the lost/received packets, it can clearly be seen that only device $3$ and device $5$ can transmit a packet and ensure that it can be decoded by at least one device upon successful reception. When only device $3$ is transmitting, the optimal solution is to target device $4$ with packet $1$ or $3$. When only device $5$ is sending, the optimal solution is to target both devices $6$ and $7$ with the packet combination $1 \oplus 3$. It can be easily shown that the optimal schedule is that both:
\begin{itemize}
\item device $3$ targets device $2$ with packet $2$.
\item device $5$ targets devices $6$ and $7$ with the combination $1 \oplus 3$.
\end{itemize}

This solution shows the high interdependence between the variables. In this section, the decoding delay problem is relaxed by focusing on cooperation without interference. In other words, the cooperation between devices is allowed only when no device experiences interference (i.e., $\mathcal{T} = \varnothing$). This limitation makes the problem more mathematically tractable by allowing the decoupling of variables, which will serve as basis to solve the original problem. Hence, this section solves the following optimization problem:

\begin{align}
&\max_{ \mathcal{A} \in \mathcal{P}(\mathcal{M}) } -|\mathcal{A} \cap M_w|- |\mathcal{S} \cap M_w| + \sum_{i \in \mathcal{A}} \left( \sum_{j \in \tau_{i}(\kappa_{i}^*) } 1-p_{ij} \right) \nonumber \\
&\text{subject to }\label{reduced} \\
&\mathcal{T} = \varnothing \nonumber \\
&\kappa_{i}^*(\mathcal{A}) = \arg \max_{\kappa_{i} \in \mathcal{P}(\mathcal{H}_{i}) } \left( \sum_{j \in \tau_{i}(\kappa_{i}) } 1-p_{ij} \right), \forall \ i \in \mathcal{A}. \nonumber
\end{align}

To solve the optimization problem \eref{reduced}, this section first proposes reformulating the problem is a more tractable form. The rest of the section illustrates solve the problem by introducing the cooperation graph.

\subsection{Relaxed Problem Formulation}

The newly introduced constraint $\mathcal{T} = \varnothing$ limits the combinations of the transmitting devices. Let $\mathcal{I}$ be the set of possible combinations of devices that satisfy the constraint. This set can be expressed as follows:
\begin{align}
\mathcal{I} = \left\{ \mathcal{A} \in \mathcal{P}(\mathcal{M}) ~\middle| ~ \mathcal{C}_i \cap \mathcal{C}_j = \varnothing, \ \forall \ (i,j) \in \mathcal{A}^2 \right\}.
\end{align}

The following theorem reformulates the optimization problem \eref{reduced}:
\begin{theorem}
The decoding delay reduction problem in an interference-less IDNC-based D2D network \eref{reduced} can be expressed as:
\begin{align}
A^* = \arg \max_{ \mathcal{A} \in \mathcal{I} } -|\mathcal{A} \cap M_w| - |\mathcal{S} \cap M_w| + \sum_{i \in \mathcal{A}} y_{i}(\kappa_{i}^*) .
\label{eq3}
\end{align}
where
\begin{align}
y_{i}(\kappa_{i}^*) = \max_{\kappa \in \mathcal{G}_{i} } \left( \sum_{j \in \tau_{i}(\kappa) } 1-p_{ij} \right). \label{eq4}
\end{align}
\label{th2}
\end{theorem}

\begin{proof}
The proof can be found in \appref{ap8}.
\end{proof}

As hinted in \eref{eq3} and \eref{eq4}, the constraint $\mathcal{T} = \varnothing$ allows the decoupling of the variables $\mathcal{A}$ and $\kappa_i, \ i \in \mathcal{A}$. Clearly, the optimization problem \eref{eq4} is equivalent to the optimization problem \eref{eq11} that can be solved using the local IDNC graph. In the next subsection, the problem of selecting the set of transmitting devices \eref{eq3} is solved using a graph theoretic formulation.

\subsection{Cooperation Graph}

To solve the optimization problem \eref{eq3}, and equivalently the optimization problem \eref{reduced}, the paper characterizes the set of feasible device combinations $\mathcal{I}$. In general, two devices can transmit simultaneously if their coverage zones are mutually disjoint. Such combinations between the devices can be represented in a graph model called herein the \emph{cooperation graph}.

The cooperation graph is constructed by creating a vertex $v_i$ for each device in the network (i.e., $\forall \ i \in \mathcal{M}$). Two vertices $v_i$ and $v_j$ are connected by an edge if their coverage zone are disjoint. In other words, they are connected if $\mathcal{C}_i \cap \mathcal{C}_j = \varnothing$. The following theorem characterizes :

\begin{theorem}
The optimal solution to the relaxed decoding delay reduction problem in an interference-less IDNC-based D2D network \eref{reduced} is the maximum weight clique problem in the cooperation graph, in which the weight of each vertex $v_{i}$ is:
\begin{align}
v_i = |\mathcal{C}_{i} \cap M_w| - |\{i\}\cap M_w| + y_{i}(\kappa_{i}^*),
\end{align}
where $\kappa_{i}^*$ is the solution to the maximum weight clique problem in the local IDNC graph illustrated in \lref{l4}.
\label{th3}
\end{theorem}

\begin{proof}
The proof can be found in \appref{ap3}.
\end{proof}

\section{Cooperative Decoding Delay Reduction}\label{sec:ext}

This section proposes to solve the decoding delay minimization problem \eref{original} by introduction clustering of devices and using the tools developed in the previous sections. The fundamental concept for solving problem \eref{original} is to generate groups of non-interfering devices (called herein clusters) and to extend the cooperative graph formulation. Each cluster can be seen as a new ``virtual" device in the network. Since the clusters are, by construction, non-interfering clusters then results of the previous section hold. Therefore, this section first shows the way of constructing such clusters then it extends the cooperative graph with the new virtual devices. Afterwards, the optimization problem \eref{original} is shown to be equivalent to a maximum weight clique problem that can be efficiently solved using existing literature from graph theory \cite{16513519,13265492}. Finally, the paper demonstrate that generating only a well-defined subset of the available clusters is sufficient to reach the optimal solution of \eref{original} with a lower complexity.

\subsection{Cluster Generation}

Let $\mathcal{A} \in \mathcal{P}(\mathcal{M})$ be a set of transmitting devices. This subsection illustrates the construction of a set of clusters $\mathcal{Z}$, called herein a clustering, that uniquely represents $\mathcal{A}$ (i.e., $\mathcal{Z}$ is a partition of $\mathcal{A}$). Each cluster $Z \in \mathcal{Z}$ (i.e., $Z \subseteq \mathcal{A}$) is seen as a vertex in the cooperation graph and finally the uniqueness of the clustering $\mathcal{Z}$ allows the reformulation of the problem in the next subsections.

Let $\mathcal{Z}$ be a clustering (i.e., a partition) of $\mathcal{A}$ such that:
\begin{enumerate}
\item All the coverage zones of clusters $Z \in \mathcal{Z}$ are pairwise disjoint.
\item Within the same cluster $Z \in \mathcal{Z}$, each subset of devices is interfering with at least another device.
\end{enumerate}
The mathematical definition of such a clustering $\mathcal{Z}$ is:
\begin{subequations}
\begin{align}
\bigoplus_{Z \in \mathcal{Z}} Z &= \mathcal{A} \label{eq14} \\
\mathcal{C}^T (Z) \cap \mathcal{C}^T (Z^\prime) &= \varnothing,\ \forall\ Z \neq Z^\prime \in \mathcal{Z} \label{eq5}\\
\mathcal{C}^T(z) \cap \mathcal{C}^T (Z \setminus z ) &\neq \varnothing,\ \forall \ z \subset Z, Z \in \mathcal{Z}, \label{eq6}
\end{align}
\end{subequations}
where $\mathcal{C}^T(X)$ is the total coverage zone of all the devices in the set $X$ defined as $\mathcal{C}^T(X) = \bigcup_{x \in X} \mathcal{C}_x$.

The following lemma states the uniqueness of the clustering $\mathcal{Z}$ for any set of transmitting devices $\mathcal{A}$.
\begin{lemma}
For any combination of transmitting devices $\mathcal{A}$, there exist a unique clustering $\mathcal{Z}$ satisfying the constraints \eref{eq14}, \eref{eq5} and \eref{eq6} simultaneously.
\label{lem1}
\end{lemma}

\begin{proof}
The proof can be found in \appref{ap4}.
\end{proof}

\subsection{Full Cooperation Graph} \label{cooper}

In this subsection, the cooperation graph is extended to include the cluster that can be seen as ``virtual" devices in the network. The fundamental concept in constructing such a graph is to preserve all the benefits of the cooperation graph while allowing all the possible combinations of transmitting devices. This goal is reached by generating particular clusters and then connecting each pair of clusters that are not interfering with each other. Let $\mathbf{Z}$ be the set of clusters satisfying the constraint \eref{eq6} defined as follows:
\begin{align}
\mathbf{Z} =& \{ Z \in \mathcal{P}(\mathcal{M}) \ | \ \mathcal{C}^T(z) \cap \mathcal{C}^T (Z \setminus z) \neq \varnothing,\ \forall \ z \subset Z\}.
\end{align}

The full cooperation graph is constructed by generating a vertex $v$ for each cluster $Z \in \mathbf{Z}$. Two vertices $v$ and $v^\prime$ representing the clusters $Z$ and $Z^\prime$ are connected if they are non-interfering clusters. In other words, they are connected if they satisfy constraint \eref{eq5}:
\begin{align}
\mathcal{C}^T (Z) \cap \mathcal{C}^T (Z^\prime) &= \varnothing.
\end{align}

Let $\mathfrak{Z}$ be the set of cliques in the full cooperation graph. The following proposition links the set of cliques to the set of transmitting devices.
\begin{proposition}
There exists a one to one mapping between the set $\mathcal{P}(\mathcal{M})$ and $\mathfrak{Z}$. In other words, for each $\mathcal{A} \in \mathcal{P}(\mathcal{M})$ there exist a unique representative $\mathcal{Z} \in \mathfrak{Z}$ and inversely.
\label{pro1}
\end{proposition}

\begin{proof}
The proof of this proposition is based on the result provided in \lref{lem1}. For a clique $\mathcal{Z} \in \mathfrak{Z}$ it is straightforward to conclude that there is a unique $\mathcal{A} \in \mathcal{P}(\mathcal{M})$ that represent it. In fact, by the construction of the vertices $Z$ they satisfy constraint \eref{eq6}. Further since $\mathcal{Z}$ is a clique then all the vertices are connected, and hence they verify constraint \eref{eq5}. For $\mathcal{A} = \bigoplus_{Z \in \mathcal{Z}} Z$ the mapping is unique. Conversely, let $\mathcal{A}$ be a set of transmitting devices. From \lref{lem1}, there exist a unique $\mathcal{Z}$ that satisfy the constraints \eref{eq14}, \eref{eq5}, and \eref{eq6}. Since all the cluster $Z \in \mathcal{Z}$ satisfy \eref{eq5} and \eref{eq6} then they are generated in the graph and they are connected. In other words, $\mathcal{Z}$ is a clique which concludes the proof.
\end{proof}

\subsection{Decoding Delay Reduction}

This subsection first reformulates the optimization problem \eref{original} in a more tractable form based on the results provided in \pref{pro1}. Afterwards, the problem is shown to be equivalent to a maximum weight clique problem in the full cooperation graph. Finally, next section provides a low complexity solving method that relies on the generation of a well-defined set of clusters to construct the cooperation graph. Given the one to one mapping between the set of clique and the set of transmitting devices, the optimization problem \eref{original} can be rewritten as follows:
\begin{subequations}
\label{reformulation}
\begin{align}
&\max_{ \mathcal{Z} \in \mathfrak{Z} } -|\mathcal{Z} \cap M_w| - |\mathcal{T} \cap M_w| - |\mathcal{S} \cap M_w| \nonumber \\
& \qquad + \sum_{Z \in \mathcal{Z}} \sum_{i \in Z}\left( \sum_{j \in \tau_{k}(\kappa_{i}^*) } 1-p_{ij} \right) \\
&\text{subject to } \kappa_{i}^*(\mathcal{Z}) = \arg \max_{\kappa_{i} \in \mathcal{P}(\mathcal{H}_{i}) } \left( \sum_{j \in \tau_{i}(\kappa_{i}) } 1-p_{ij} \right)
\end{align}
\end{subequations}

The following theorem characterizes the global optimal solution to the decoding delay reduction in IDNC-based D2D network \eref{original}:
\begin{theorem}
The optimal solution to the decoding delay reduction problem in partially connected IDNC-based D2D network \eref{original} is the maximum weight clique problem in the cooperation graph, in which the weight of each vertex $v$ representing the cluster $Z$ is:
\begin{align}
v &= |\mathcal{C}^T(Z) \cap M_w| - | Z \cap M_w| - |\mathcal{T}(Z) \cap M_w| \nonumber \\
& + \sum_{i \in Z}\left( \sum_{j \in \tau_{i}(\kappa_{i}^*) } 1-p_{ij} \right) ,
\end{align}
where $\kappa_{k}^*$ is the maximum weight clique problem in the local IDNC graph for packet generation:
\begin{align}
\kappa_{i}^* = \arg \max_{\kappa \in \mathcal{G}_i} \left( \sum_{j \in \tau_{i}(\kappa) } 1-p_{ij} \right).
\end{align}
\label{th4}
\end{theorem}

\begin{proof}
The proof can be found in \appref{ap5}.
\end{proof}

\subsection{Low Complexity Solving Method}

This subsection presents a low complexity for solving the decoding delay minimization problem \eref{original}. As shown in the previous paragraph, the optimal solution can be represented as a maximum weight clique in the cooperation graph. This subsection answers to the question: Are all the clusters $Z \in \mathbf{Z}$ needed in the generation of the cooperation graph to achieve the optimal solution to the optimization problem? The paper answers the question by introducing a sequential low complexity method for constructing the cooperation graph by eliminating the vertices that are surely not part of the maximum weight clique.

Let $Z \in \mathbf{Z}$ be a cluster and let $k \notin Z$ be a device such that:
\begin{align}
\mathcal{C}_k \cap \mathcal{C}^T(Z) \neq \varnothing.
\label{eq16}
\end{align}

Let $\mathcal{Z} \in \mathfrak{Z}$ be a clustering (i.e., a clique in the cooperation graph) such that $Z \in \mathcal{Z}$. First note that $k \notin Z^\prime, \ \forall \ Z^\prime \in \mathcal{Z}$. Otherwise, since $k$ satisfy \eref{eq16} then $\mathcal{Z}$ would violate \eref{eq5} and hence it is not a clique. This subsection first exhibits the condition under which the consideration of clusters that include device $k$ would be beneficial. Afterwards, it proposes a method for generating only such clusters.

Let $\mathcal{Z}^k$ be a clustering, similar to $\mathcal{Z}$, in which the cluster $Z$ is replaced by the cluster $Z^k=Z \cup \{k\}$. In other words, the clustering $\mathcal{Z}^k$ can be written as $\mathcal{Z}^k =(\mathcal{Z} \setminus Z) \cup \{Z \cup \{k\}\}$. In order to determine if all the clusters are needed to achieve the optimal solution, the paper compute the different delays of $\mathcal{Z}$ and $\mathcal{Z}^k$. The delay for a set of transmitting users $\mathcal{Z}$ can be expressed as the delay induced by $Z$ and the delays induced by $Z^\prime \neq Z$ as illustrated in \eref{reformulation}. Therefore, the difference in delay between $\mathcal{Z}$ and $\mathcal{Z}^k$ is the following:
\begin{align}
\Delta &= |\mathcal{C}^T(Z) \cap M_w| - |Z \cap M_w| - |\mathcal{T}(Z) \cap M_w| \nonumber \\
& \qquad + y(Z) - \bigg(|\mathcal{C}^T(Z^k) \cap M_w| - |Z^k \cap M_w| \nonumber \\
& \qquad - |\mathcal{T}(Z^k) \cap M_w| + y(Z^k) \bigg) \label{eq10} \\
&= -|(\mathcal{C}^T(Z^k) \setminus \mathcal{C}^T(Z)) \cap M_w| + |\{k\} \cap M_w| \nonumber \\
& \qquad + |(\mathcal{T}(Z^k) \setminus \mathcal{T}(Z)) \cap M_w | + y(Z) - y(Z^k) \nonumber.
\end{align}
where
\begin{align}
y(Z) = \sum_{i \in Z} \left[ \max_{\kappa_i \in \mathcal{G}_i} \left(\sum_{j \in \tau_i(\kappa_i)} (1 - p_{ij} ) \right)  \right]
\label{eq9}
\end{align}

Clearly, if the quantity $\Delta$ defined in \eref{eq10} is a positive number, then including device $k$ in the clique do not provide any gain as it brings more interference than it serves devices. This concludes that such clustering $\mathcal{Z}^k$ is not the maximum weight clique since $\mathcal{Z}$ is a clique with a higher weight. Based on this observation, the rest of this subsection provide an efficient method for constructing the cooperation graph while avoiding, at maximum, generating the clusters that are surely not part of the maximum weight clique.

To construct the cooperation graph, first generate a vertex $v^0_j$ for each device $j \in \mathcal{M}$ in the network. Let such set of vertices be called the first layer of the graph. Two vertices $v_i^0$ and $v_j^0$ respecting the non-interference condition $\mathcal{C}^T(v^0_j) \cap \mathcal{C}^T(v^0_i) = \varnothing$ are connected. Afterwards, the weight of each vertex $v$, representing the cluster $Z$, is computed as follows:
\begin{align}
w(v) = |\mathcal{C}^T \cap M_w| - |Z \cap M_w| - |\mathcal{T} \cap M_w| + y(Z),
\end{align}
where the function $y(.)$ is defined in \eref{eq9}.

After this first step, the second layer of the graph is constructed by merging the clusters of the previous layer with the cluster of the first layer if the equality in \eref{eq10} holds. In other words, for each pair of clusters $v^0_j$ and $v^0_i$ that are not connected (otherwise the clustering violates \eref{eq6}), two scenarios can be distinguished:
\begin{itemize}
\item $w(\{v^0_i,v^0_j\}) \geq \max (v^0_i,v^0_j)$: generate cluster $v^1_{ij}$ representing the cluster $Z=\{i,j\}$.
\item $w(\{v^0_i,v^0_j\}) < \max (v^0_i,v^0_j)$: the cooperation is not beneficial. The cluster is not created.
\end{itemize}

After the cluster generating phase, the newly created vertices are connected with the other ones that satisfy the non-interference constraint. The process is repeated for all layers of the graph. Hence for layer $k+1$, combine all vertices $v^0_i$ with vertices $v^k_j$ if they are not connected and compute the weight of $w(\{v^0_i,v^k_j\})$ to decide to generate or not the cluster. The steps of the algorithm are summarized in \algref{algo2}.

\begin{algorithm}[t]
\begin{algorithmic}
\REQUIRE $\mathcal{C}_{i}, \forall\ i \in \mathcal{M}$.
\STATE Initialize $k = 0$.
\STATE Construct $\mathcal{G}^0(\mathcal{V},\mathcal{E})$ using section
\STATE Initialize $t = \TRUE$.
\WHILE{$t = \TRUE$}
\STATE Set $t \leftarrow \FALSE$.
\STATE Initialize $\mathcal{G}^{k+1} = \varnothing$
\FORALL{ $v^{k} \in \mathcal{G}^k$}
\FORALL{ $z^0 \in \mathcal{G}^0$}
\IF {$w(\{v^{k},z^0\}) \geq \max (w(v^{k}),w(z^0)) $}
\STATE Set $t \leftarrow \TRUE$.
\STATE Create vertex $y^{k+1} = \{v^k,z^0\}$
\STATE Set $\mathcal{G}^{k+1} \leftarrow \mathcal{G}^{k+1} \cup y$
\ENDIF
\ENDFOR
\ENDFOR
\STATE Set $k \leftarrow k + 1$.
\FORALL{ $y^{k} \in \mathcal{G}^{k}$}
\FORALL{ $x \in \bigcup_{i=0}^k\mathcal{G}^k$}
\IF {$\mathcal{C}^T(x) \cap \mathcal{C}^T(y^{k}) = \varnothing$ }
\STATE Create edge between $y^{k}$ and $x$
\ENDIF
\ENDFOR
\ENDFOR
\ENDWHILE
\STATE Set $\mathcal{G}=\bigcup_{i=0}^k \mathcal{G}^k$
\end{algorithmic}
\caption{Graph generation}
\label{algo2}
\end{algorithm}

\begin{theorem}
Solving the maximum weight clique in the cooperation graph generated by \algref{algo2} yields the optimal solution to the optimization problem \eref{original}. In other words, the maximum weight clique in the graph produced by \algref{algo2} is the same as the maximum weight clique in the cooperation graph constructed in \sref{cooper}.
\label{th5}
\end{theorem}

\begin{proof}
The proof can be found in \appref{ap6}.
\end{proof}

\section{Simulations Results}\label{sec:sim}

In this section, the simulation results comparing the decoding delay performance of the proposed schemes in the partially connected D2D system (problem \eref{original} and the heuristic interference-less solution \eref{reduced}) to both the conventional PMP system \cite{5683677} (base station is responsible for both the \emph{initial} phase and the \emph{recovery} phase) and the solution proposed in \cite{6620795} for fully connected network.

In these simulations, the decoding delay is computed over a large number of realizations, and the mean value is presented. Since the short range communications are more reliable than the base-station-to-device communications \cite{6620795,6404743}, the D2D packet erasure probability $P$ is set to be $P = Q/2$ where $Q$ is the base-station-to-device packet erasure probability. These erasure probabilities are assumed to be perfectly known for all devices.

\begin{figure}[t]
\centering
\includegraphics[width=1\linewidth]{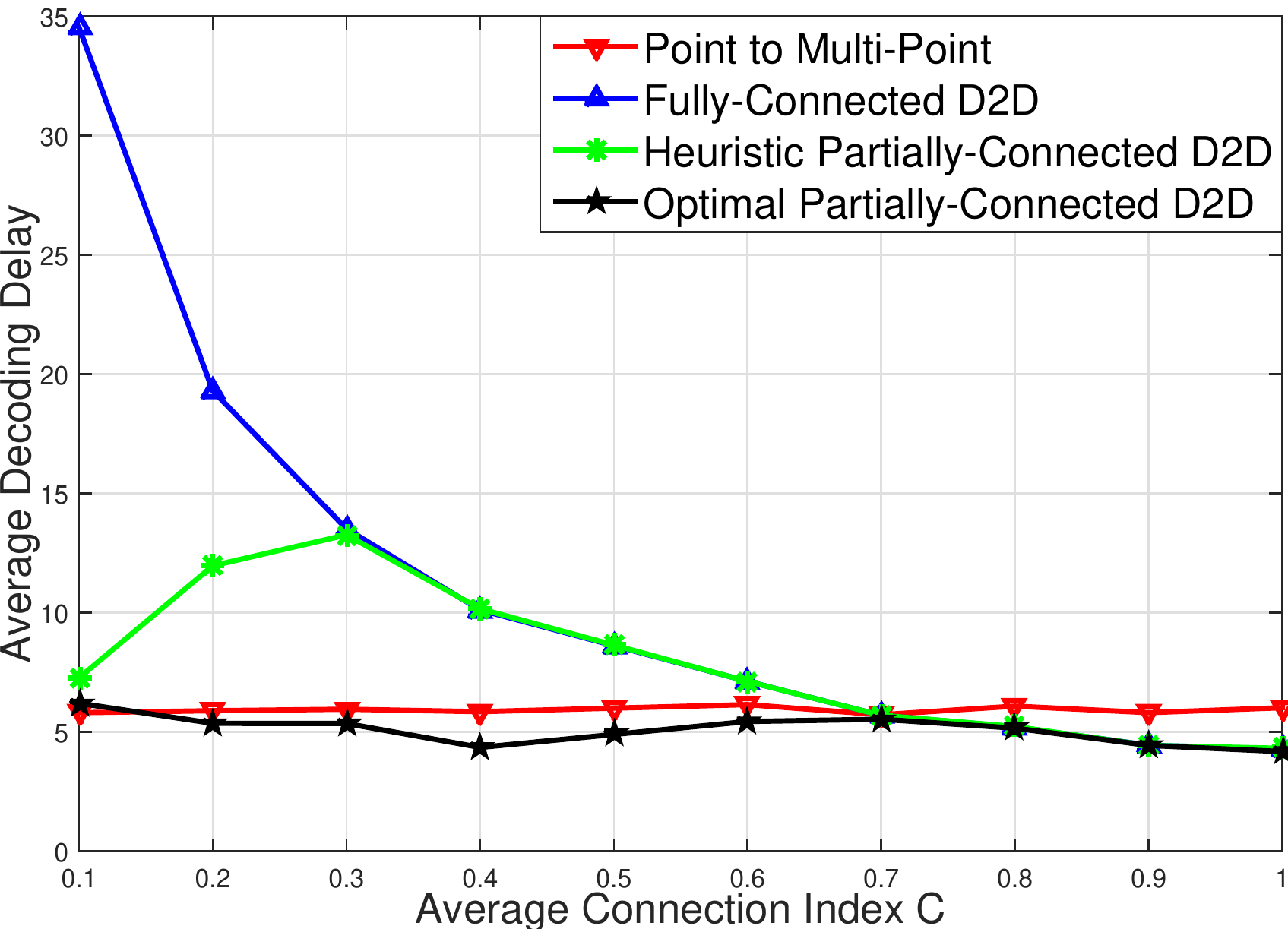}\\
\caption{Mean decoding delay versus the connectivity index $C$ for a network composed of $M=60$ devices, $N=30$ packets, an erasure probability $P=0.1$, and $Q=0.2$.}\label{fig:C}
\end{figure}

\begin{figure}[t]
\centering
\includegraphics[width=1\linewidth]{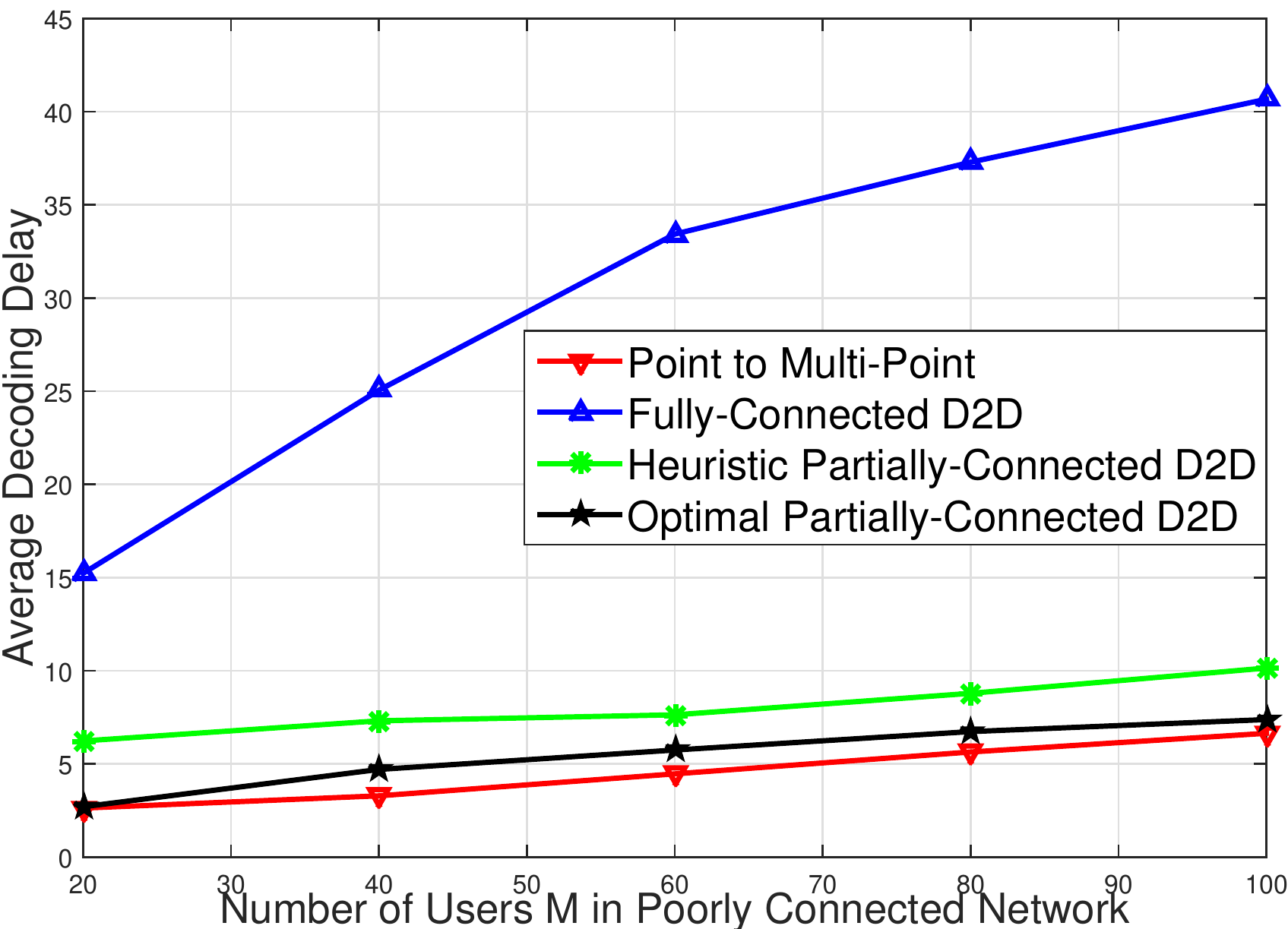}\\
\caption{Mean decoding delay versus number of devices $M$ for a network composed of $N=30$ packets, a connecitvity index $C=0.1$ an erasure probability $P=0.1$, and $Q=0.2$}\label{fig:M1}
\end{figure}

\begin{figure}[t]
\centering
\includegraphics[width=1\linewidth]{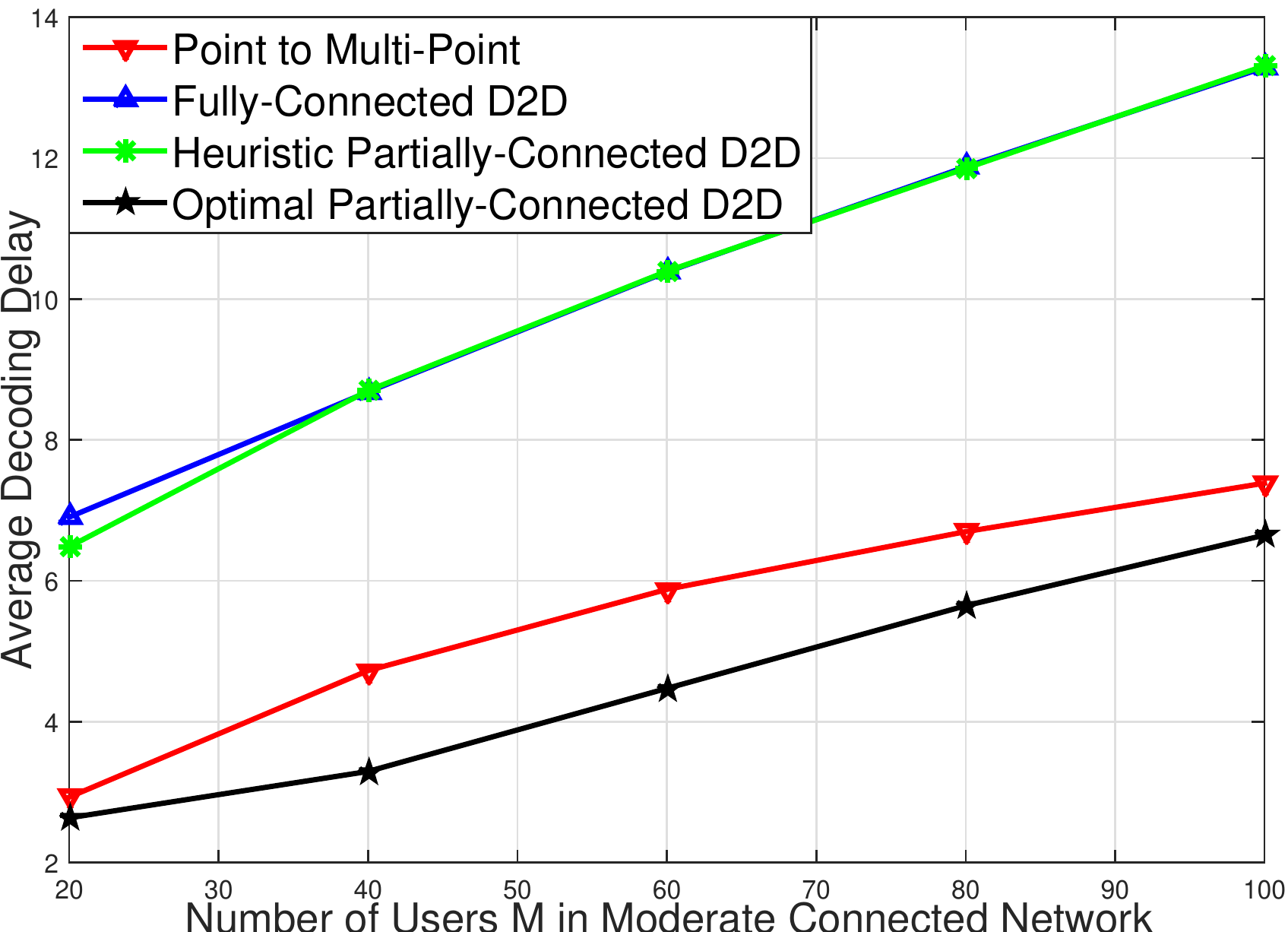}\\
\caption{Mean decoding delay versus number of devices $M$ for a network composed of $N=30$ packets, a connecitvity index $C=0.4$ an erasure probability $P=0.1$, and $Q=0.2$.}\label{fig:M4}
\end{figure}

\begin{figure}[t]
\centering
\includegraphics[width=1\linewidth]{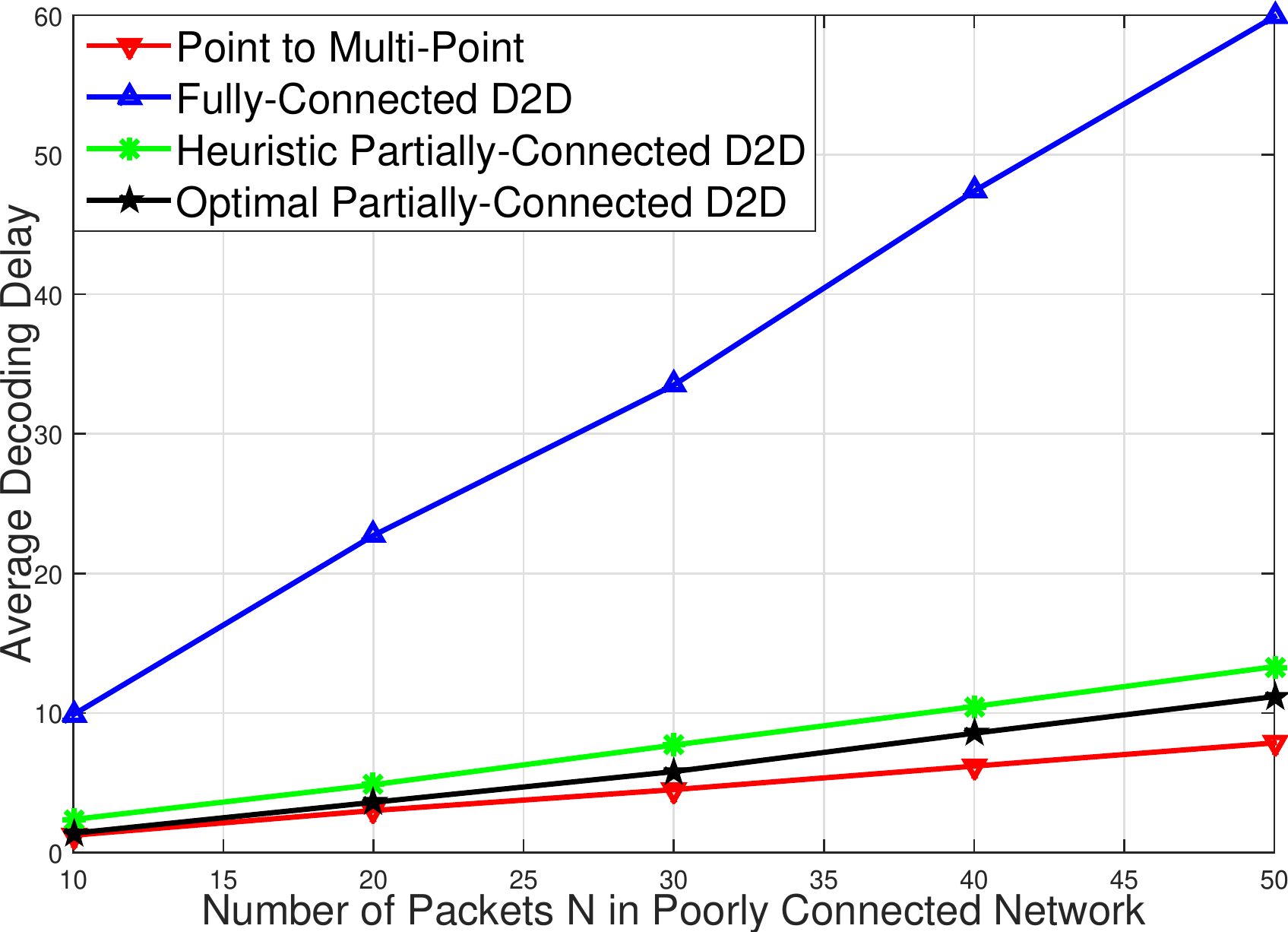}\\
\caption{Mean decoding delay versus number of packets $N$ for a network composed of $M=60$ devices, a connecitvity index $C=0.1$ an erasure probability $P=0.1$, and $Q=0.2$.}\label{fig:N1}
\end{figure}

\begin{figure}[t]
\centering
\includegraphics[width=1\linewidth]{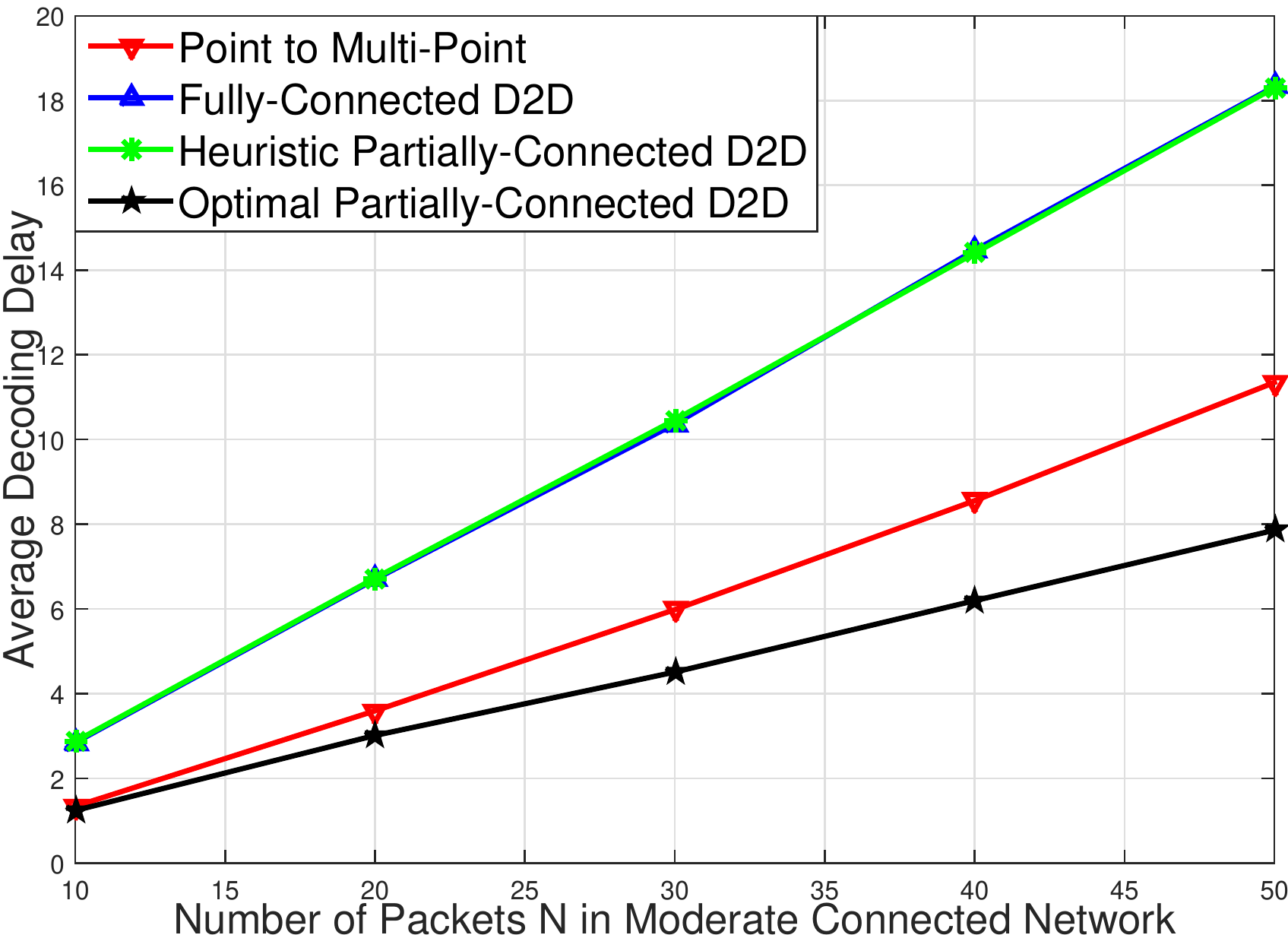}\\
\caption{Mean decoding delay versus number of packets $N$ for a network composed of $M=60$ devices, a connecitvity index $C=0.4$ an erasure probability $P=0.1$, and $Q=0.2$.}\label{fig:N4}
\end{figure}

\begin{figure}[t]
\centering
\includegraphics[width=1\linewidth]{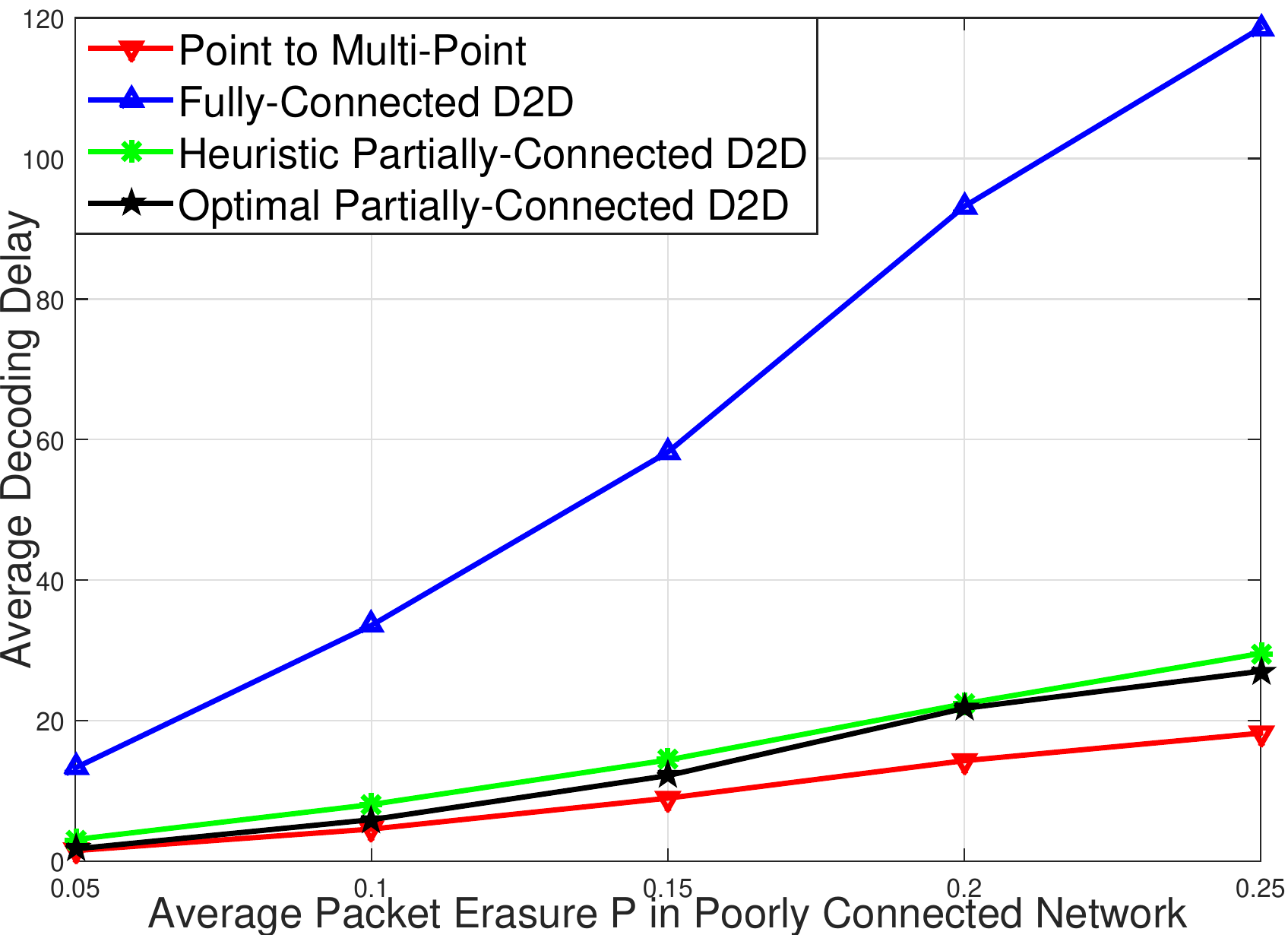}\\
\caption{Mean decoding delay versus erasure probability $P$ for a network composed of $M=60$ devices, $N=30$ packets, a connecitvity index $C=0.1$ an erasure probability $Q=2P$.}\label{fig:P1}
\end{figure}

\begin{figure}[t]
\centering
\includegraphics[width=1\linewidth]{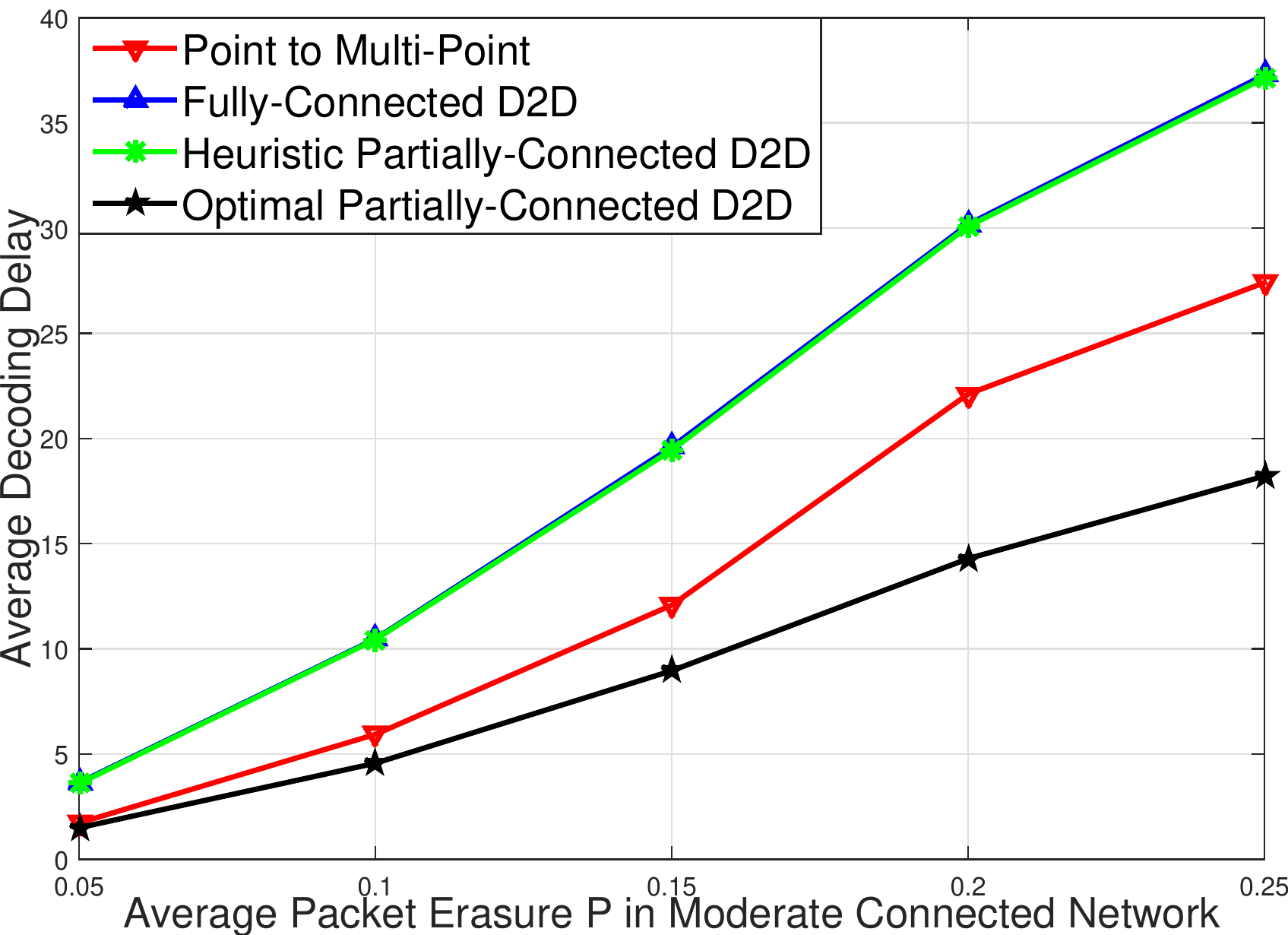}\\
\caption{Mean decoding delay versus erasure probability $P$ for a network composed of $M=60$ devices, $N=30$ packets, a connecitvity index $C=0.4$ an erasure probability $Q=2P$.}\label{fig:P4}
\end{figure}

\fref{fig:C} depicts the comparison of the average decoding delay achieved by the PMP configuration, the fully connected D2D policy, and our partially connected D2D optimal nad heuristic policies against the connectivity index $C$ for $M=60$, $N=30$, $P=0.1$ and $Q=0.2$. From \fref{fig:C}, we clearly see that for a low connectivity index ($C \leq 0.3$) our partially connected approach largely outperforms the fully connected approach. This can be explained by the fact that in the fully connected D2D, only one device is transmitting at a time. For a highly connected graph, a single device can targeted all the others since the number of devices out of the transmission range is negligible. Whereas for a low connectivity index, this term is no longer negligible. Further, the heuristic solution clearly degrades as the connectivity index increase. This can be explained by the fact that as the connectivity index increases, the number of devices that can transmit simultaneously decreases and thus the solution gets closer and closer to the fully connected D2D.

\fref{fig:M1} and \fref{fig:M4} illustrate the average decoding delay achieved by the same policies against the number of devices $M$ for $N=30$, $P=0.1$, $Q=0.2$ for a poorly connected network ($C=0.1$ in \fref{fig:M1}) and a moderately connected network ($C=0.4$ in \fref{fig:M4}). \fref{fig:N1} and \fref{fig:N4} show a similar comparison against the number of packets $N$ for $M=60$, $P=0.1$, $Q=0.2$ and $C=0.8$ for a poorly connected network ($C=0.1$ in \fref{fig:N1}) and a moderately connected network ($C=0.4$ in \fref{fig:N4}).

For connectivity index $C = 0.1$, the heuristic partially connected and the optimal partially connected policies have the same performance as displayed in \fref{fig:M1} and \fref{fig:N1}. This can be explained by the fact that at low connectivity, the set of devices that can transmit simultaneously while respecting the non-interference constraint is large and thus allowing interference do not provide a significant gain. Further, the partially connected algorithms provide an appreciable gain as compared with the fully connected solution as the number of devices and packets increases. For $C=0.4$, we can clearly from \fref{fig:M4} and \fref{fig:N4} see that the optimal partially connected policy outperforms the PMP approach. Note that the PMP approach does not suffer from delay due to devices out of the transmitting range and delay encountered by the transmitting devices. As the connectivity index increases, these additional decoding delay decreases, which explain that the optimal D2D approach outperforms the PMP at these connectivity indexes.

\fref{fig:P1} and \fref{fig:P4} depict the same comparison for the decoding delay against the packet erasure probability $P$ for $M=60$, $N=30$, $Q=2P$, $C=0.1$ for \fref{fig:P1}, and $C=0.4$ for \fref{fig:P4}. For a low connectivity index $C$ our proposed solutions outperform the fully connected one for all values of erasure probabilities as shown in \fref{fig:P1}. As the connectivity index increases (\fref{fig:P4}), we clearly see that the performance of the optimal partially connected D2D is better than the PMP one. This can be explained by the fact that in our simulation, we fixed the base station-to-device erasure probability to be twice as large as the device-to-device erasure probability. For a small erasure probability, the different between the D2D and the PMP erasure is not significant. As $P$ increases, $Q$ increases twice as much, which explain that the D2D approaches the PMP performance.

\section{Conclusion}\label{sec:conc}

The paper investigates the decoding delay reduction problem in IDNC-based device-to-device communications network. In a D2D configuration, devices exchange encoded packets to help hasten the recovery of the lost packets of other devices in their transmission range. This paper considers a partially connected network in which more than one device can simultaneously transmit. The decoding delay minimization problem is formulated as a joint optimization over the set of transmitting devices and the packet combination to be transmitted. The paper first present a heuristic solution in which only non-interfering devices are allowed to transmit simultaneously. Afterwards, the paper proposes the global optimal solution to the decoding delay using a clustering approach. Both problems are shown to be equivalent to a maximum weight clique problem in the constructed cooperation graph. Through extensive simulations, the decoding delay experienced by all devices in the Point to Multi-Point (PMP) configuration, the fully connected D2D (FC-D2D) configuration and the more practical partially connected D2D (PC-D2D) configuration are compared. Numerical results suggest that the PC-D2D outperforms the FC-D2D and provides appreciable gain especially for poorly connected networks.

\bibliographystyle{IEEEtran}
\bibliography{references}

\appendices

\numberwithin{equation}{section}

\section{Proof of \thref{th1}} \label{ap7}

In order to formulate the decoding delay minimization problem in partially connected IDNC-based D2D network, the expected decoding delay increase of each device is first expressed. Afterward, these expressions are used to formulate the problem as a joint optimization over the set of transmitting devices and the packet combinations to be transmitted.

From the definition of the decoding delay, device $i$, with non-empty Wants sets, experiences one unit increase of decoding delay if he cannot hear exactly one transmission or if he hears a packet that is either non-instantly decodable or non-innovative. The event of a device $i$, with non-empty Wants set, not be able to hear exactly one transmission occurs in the following scenarios:
\begin{itemize}
\item The device cannot hear any transmission: This event occurs if one of the following is true:
\begin{itemize}
\item The device is transmitting (i.e., $i \in \mathcal{A} \cap M_w$).
\item The device is out of the transmission range of the transmitting devices (i.e., $i \in \mathcal{S} \cap M_w $).
\end{itemize}
\item The device can hear multiple transmissions: This event occurs when the device is in the interference region of the transmitting devices (i.e., $i \in \mathcal{T} \cap M_w$).
\end{itemize}

The event of a device $i$, with non-empty Wants set, hearing a packet that is either non-instantly decodable or non-innovative happens when the device is in the opportunity zone of one of the transmitting devices (i.e., $i \in \mathcal{O}_j, j \in \mathcal{A}$) and one of the following event is true:
\begin{itemize}
\item He receives a non-innovative packet combination $\kappa_{j}$.
\item He receives a non-decodable packet combination $\kappa_{j}$.
\end{itemize}
These last two events translate the fact that device $i$ is not targeted by the transmission (i.e., $i \in (\mathcal{O}_{j} \setminus \tau_{j}(\kappa_{j}))$. Due to the erasure nature of the links, the probability that the event occurs is $1-p{ji}$.

In order to derive the expected decoding delay increase for an arbitrary device $i$ in the network, we first partition the set of devices. From \lref{l2} (\appref{app10}), for any set of transmitting devices $\mathcal{A}$, the following sets forms a partition of the set of devices $\mathcal{M}$:
\begin{align}
&\mathcal{M} = \bigoplus_{i \in \mathcal{A}} \left( \mathcal{O}_{i} \cap M_w \right) \oplus \left( \mathcal{A} \cap M_w \right) \nonumber \\
& \qquad \qquad \oplus \left( \mathcal{T} \cap M_w \right) \oplus \left( \mathcal{S} \cap M_w \right) \oplus \overline{M}_w.
\label{eqth1}
\end{align}

Given the partition of devices in \eref{eqth1} and the analysis above, the expected individual decoding delay $d_j(\mathcal{A},\kappa)$ of device $j$ when devices in $\mathcal{A}$ are transmitting respectively the packets combinations $\kappa_{i}$ and $\kappa = \{\kappa_{i}\}_{i \in \mathcal{A}}$ is given by:
\begin{align}
&\mathds{P}(d_j(\mathcal{A},\kappa)=1) =
\begin{cases}
0 &\text{ if } j \in \overline{M}_w \\
1 &\text{ if } j \in \mathcal{A} \cap M_w \\
1 &\text{ if } j \in \mathcal{T} \cap M_w \\
1 &\text{ if } j \in \mathcal{S} \cap M_w \\
0 &\text{ if } j \in (\mathcal{O}_{i} \cap M_w) \cap \tau_{i}(\kappa_{i}) \\
1-p_{ij} &\text{ if } j \in (\mathcal{O}_{i} \cap M_w) \setminus \tau_{i}(\kappa_{i}).
\end{cases} \nonumber
\end{align}

Let $D(\mathcal{A},\kappa)$ be the total decoding delay increase. From \eref{eqth1}, all the sets form a partition of $\mathcal{M}$ and thus are disjoint. Therefore, the expected overall decoding delay can be written as:
\begin{align}
\mathds{E}\left[D(\mathcal{A},\kappa)\right] &= \sum_{i \in \mathcal{M}} d_i(\mathcal{A},\kappa) \nonumber \\
&= |\mathcal{A} \cap M_w| + |\mathcal{T} \cap M_w| + |\mathcal{S} \cap M_w| \nonumber \\
& \qquad + \sum_{i \in \mathcal{A}} \left( \sum_{j \in (\mathcal{O}_{i} \cap M_w) \setminus \tau_{i}(\kappa_{i}) } 1-p_{ij} \right) .
\label{eqth2}
\end{align}

Using the expected decoding delay increase expression \eref{eqth2}, the decoding delay minimization problem can be formulated as the following joint optimization over the set of transmitting devices and the packet combination to be transmitted:
\begin{align}
&\underset{\kappa_{i}(\mathcal{A}) \in \mathcal{P}(\mathcal{H}_{i})}{\min_{ \mathcal{A} \in \mathcal{P}(\mathcal{M}) }} \mathds{E}\left[D(\mathcal{A},\kappa)\right] \nonumber \\
&= \underset{\kappa_{i} \in \mathcal{P}(\mathcal{H}_{i})}{\min_{ \mathcal{A} \in \mathcal{P}(\mathcal{M}) }} |\mathcal{A} \cap M_w| + |\mathcal{T} \cap M_w| \nonumber \\
& \qquad + |\mathcal{S} \cap M_w| + \sum_{i \in \mathcal{A}} \left( \sum_{j \in (\mathcal{O}_{i} \cap M_w) \setminus \tau_{i}(\kappa_{i}) } 1-p_{ij} \right) \nonumber \\
&= \underset{\kappa_{i} \in \mathcal{P}(\mathcal{H}_{i})}{\max_{ \mathcal{A} \in \mathcal{P}(\mathcal{M}) }} -|\mathcal{A} \cap M_w| - |\mathcal{T} \cap M_w| \nonumber \\
& \qquad - |\mathcal{S} \cap M_w| + \sum_{i \in \mathcal{A}} \left( \sum_{j \in \mathcal{O}_{i} \cap \tau_{i}(\kappa_{i}) } 1-p_{ij} \right) \nonumber \\
&= \underset{\kappa_{i} \in \mathcal{P}(\mathcal{H}_{i})}{\max_{ \mathcal{A} \in \mathcal{P}(\mathcal{M}) }} -|\mathcal{A} \cap M_w| - |\mathcal{T} \cap M_w| \nonumber \\
& \qquad - |\mathcal{S} \cap M_w| + \sum_{i \in \mathcal{A}} \left( \sum_{j \in \tau_{i}(\kappa_{i}) } 1-p_{ij} \right).
\label{eqth3}
\end{align}

Apparently, from the problem formulation \eref{eqth3}, the packet combination that the transmitting devices can generate only affects the last term of the expression. Therefore, the decoding delay minimization problem can be reformulated as follows:
\begin{align}
&\max_{ \mathcal{A} \in \mathcal{P}(\mathcal{M}) } -|\mathcal{A} \cap M_w| - |\mathcal{T} \cap M_w| - |\mathcal{S} \cap M_w| \nonumber \\
& \qquad + \sum_{i \in \mathcal{A}} \left( \sum_{j \in \tau_{i}(\kappa_{i}^*) } 1-p_{ij} \right) \\
&\text{subject to } \nonumber\\
&\kappa_{i}^*(\mathcal{A}) = \arg \max_{\kappa_{i} \in \mathcal{P}(\mathcal{H}_{i}) } \left( \sum_{j \in \tau_{i}(\kappa_{i}) } 1-p_{ij} \right), \ \forall \ i \in \mathcal{A}.
\end{align}

\section{Proof of \lref{l4}} \label{ap9}

To proof this lemma, an approach similar to the one employed in \cite{6570827} for a PMP network is used. Since a device $i$ can target only the devices in his opportunity zone $\mathcal{O}_i$ and can make packet combination only using the packets it already holds, then the network can be reduced. In other words, considering the reduced network consisting of a set $\mathcal{M}^\prime = \mathcal{O}_i$ of devices and a set $\mathcal{N}^\prime = \mathcal{H}_i$ of packets yields the same solution. In fact, since $\tau_{i}(\kappa_{i}) \subseteq \mathcal{O}_i$, then the targeted devices in the original network are the same as in the reduced one. Moreover, given $\kappa_{i} \subseteq \mathcal{H}_i$ then the packet combination remains unchanged. Therefore, the problem can be expressed as:
\begin{align}
\kappa_{i}^* &= \arg \max_{\kappa_{i} \in \mathcal{P}(\mathcal{N}^\prime) } \left( \sum_{j \in \tau_{i}(\kappa_{i}) } 1-p_{ij} \right).
\label{eq:15}
\end{align}

According to the analysis done in \cite{6570827}, the optimization problem \eref{eq:15} is equivalent to a maximum weight clique in the IDNC graph \cite{5683677} of the reduced network. Since the IDNC graph of the reduced network is equivalent to the local IDNC graph introduced in \sref{sec:prob}, the result can be easily extended to the local IDNC graph. Therefore, the set of all feasible packet combinations in local IDNC is represented by all maximal cliques in $\mathcal{G}_i$. To generate a packet combination, binary XOR is applied to all the packets identified by the vertices of a selected maximal clique $\kappa$ in $\mathcal{G}_i$. The targeted devices by this transmission $\kappa$ are those determined by the vertices of the selected maximal clique. In terms of the local IDNC graph, the optimization problem \eref{eq11} is equivalent to:
\begin{align}
\kappa_{i}^* &= \arg \max_{\kappa_{i} \in \mathcal{P}(\mathcal{H}_{i}) } \left( \sum_{j \in \tau_{i}(\kappa_{i}) } 1-p_{ij} \right) \nonumber \\
&= \arg \max_{\kappa_i \in \mathcal{G}_{i} } \left( \sum_{j \in \tau_{i}(\kappa) } 1-p_{ij} \right), \ \forall \ i \in \mathcal{A}.
\end{align}
where $\kappa$ is a clique in the local IDNC graph. Therefore, the solution of the optimization problem \eref{eq11} for device $i \in \mathcal{A}$ is the maximum weight clique in the local IDNC graph $\mathcal{G}_i$ in which the weight of each vertex $v_{jl}$ is $1-p_{ij}$.

\section{Proof of \thref{th2}} \label{ap8}

To show the result, we first prove that the local IDNC graph is independent of the set of transmitting devices. This result, as shown in this proof, allows the decoupling of the variable of the problem \eref{original}, i.e., the decoupling of the problems \eref{eq12} and \eref{eq11}.

From \lref{l3} (\appref{app10}), under the non-interference condition $\mathcal{T} = \varnothing$, the opportunity zone of any transmitting device $i$ is independent of the set of transmitting devices. In other words, for any set of transmitting device $\mathcal{A} \in \mathcal{I}$:
\begin{align}
\mathcal{O}_i(\mathcal{A}) = \mathcal{O}_i = \mathcal{C}_i \setminus \{i\}.
\label{eqth21}
\end{align}

The only variable that depends on the set of transmitting devices $\mathcal{A}$, in the construction of the local IDNC graph of device $i$, is the opportunity zone $\mathcal{O}_i$. Therefore, given that \eref{eqth21} holds under the non-interference constraint, then the local IDNC graph does not depend on the transmitting devices (i.e., $\mathcal{G}_i(\mathcal{A}) = \mathcal{G}_i$). As a consequence, the optimal packet combination device $i$ can generate (optimization problem \eref{eq11}) do not depend on the set of transmitting devices as showed in following equation:
\begin{align}
\kappa_{i}^*(\mathcal{A}) &= \arg \max_{\kappa \in \mathcal{G}_{i}(\mathcal{A}) } \left( \sum_{j \in \tau_{i}(\kappa) } 1-p_{ij} \right) \nonumber \\
&= \arg \max_{\kappa \in \mathcal{G}_{i} } \left( \sum_{j \in \tau_{i}(\kappa) } 1-p_{ij} \right) = \kappa_{i}^* ,\ \forall \ i \in \mathcal{A}.
\label{eq1}
\end{align}

Let $y_{i}(\kappa_{i}^*)$ be the value of the objective function of \eref{eq1}. Note that $y_{i}(\kappa_{i}^*)$ is a non-negative function. The problem of selecting the transmitting devices can be expressed as:
\begin{align}
A^* &= \arg \max_{ \mathcal{A} \in \mathcal{I} } -|\mathcal{A} \cap M_w| - |\mathcal{S} \cap M_w| + \sum_{i \in \mathcal{A}} y_{i}(\kappa_{i}^*) . \nonumber
\end{align}

\section{Proof of \thref{th3}}\label{ap3}

To prove this theorem, we first reformulate the optimization problem \eref{eq3} in a more tractable form. Afterward, we show that there is a one to one mapping between the set of feasible transmitting devices and the cliques in the cooperation graph. Finally, since the weight of the clique is equivalent to the objective function of \eref{eq3} then the optimal solution of \eref{eq3} is the maximum weight clique in the cooperation graph.

Let $ \mathcal{C}^T(\mathcal{A}) = \bigcup\limits_{i \in \mathcal{A}} \mathcal{C}_{i}$ be the total coverage zone of the transmitting devices. The set of devices, with non-empty Wants set, out of the transmission range of the transmitting devices can be written, under the non-interference constraint, as follows:
\begin{align}
\mathcal{S} \cap M_w &= \left( \mathcal{M} \setminus \mathcal{C}^T(\mathcal{A}) \right) \cap M_w = \left( \mathcal{M} \cap \overline{\mathcal{C}^T}(\mathcal{A}) \right) \cap M_w \nonumber \\
&= \overline{\mathcal{C}^T}(\mathcal{A}) \cap M_w = \left( \overline{\mathcal{C}^T}(\mathcal{A}) \cap M_w \right) \cup \left( M_w \cap \overline{M}_w \right) \nonumber \\
&= M_w \cap \left( \overline{\mathcal{C}^T}(\mathcal{A}) \cup \overline{M}_w \right) = M_w \cap \left( \overline{\mathcal{C}^T(\mathcal{A}) \cap M_w } \right) \nonumber \\
&= M_w \setminus \left( \mathcal{C}^T(\mathcal{A}) \cap M_w \right).
\label{eq8}
\end{align}
Since $\mathcal{T} = \varnothing$, then we have $\mathcal{C}^T(\mathcal{A}) = \bigoplus\limits_{i \in \mathcal{A}} \mathcal{C}_{i}$. Therefore, the cardinality of the set $\mathcal{S} \cap M_w$ can be written as follows:
\begin{align}
|\mathcal{S} \cap M_w| &= |M_w| - \left| \left( \bigoplus\limits_{i \in \mathcal{A}} \mathcal{C}_{i} \cap M_w \right) \right| \nonumber \\
&= |M_w| - \left| \bigoplus\limits_{i \in \mathcal{A}} \left( \mathcal{C}_{i} \cap M_w \right) \right| \nonumber \\
&= |M_w| - \sum_{i \in \mathcal{A}} | \mathcal{C}_{i} \cap M_w |.
\label{eq2}
\end{align}

Note that, the first term in \eref{eq2} is constant with respect to the set of transmitting devices and the transmitted packet combinations. Thus, this term is be ignored in the optimization problem. Clearly, we have $|\mathcal{A} \cap M_w | = \sum_{i \in \mathcal{A}} |\{i\} \cap M_w|$. Hence, the optimization problem \eref{eq3} can be reformulated as follows:
\begin{align}
A^* &= \arg \max_{ \mathcal{A} \in \mathcal{I} } \sum_{i \in \mathcal{A}} \left( |\mathcal{C}_{i} \cap M_w| - |\{i\}\cap M_w| + y_{i}(\kappa_{i}^*) \right)
\label{eqth31}
\end{align}

The rest of the section shown that there is a one to one mapping between the set of feasible transmitting devices (i.e., $\mathcal{A} in \mathcal{I}$) and the set of cliques in the cooperation graph. Finally, to conclude the proof, we note that the weight of the clique is equivalent to the objective function of \eref{eqth31}.
We now prove that the solution of this problem is equivalent to the maximum weight clique in the cooperation graph where the weight of each vertex is:

Let $\mathcal{A} \in \mathcal{I}$ be any combination of transmitting devices satisfying the non-interference constrain. We show that it can be represented by a clique $\kappa$ in the cooperation graph. Let $v_i$ be the vertices associated with the devices in $, \ i \in \mathcal{A}$. By definition of $\mathcal{I}$ we have $\mathcal{C}_{i} \cap \mathcal{C}_{j} = \varnothing,\ \forall\ i \neq j \in \mathcal{A}$. Therefore, vertices $v_i$ and $v_j$ are connected which concludes that $\kappa$ is a clique in the cooperation graph.

In a similar way, let $\kappa$ be a clique in the cooperation graph associated with the set of transmitting devices $\mathcal{A}$. Since all the nodes are pairwise connected then $\mathcal{C}_{i} \cap \mathcal{C}_{j} = \varnothing,\ \forall\ v_i \neq v_j \in \kappa$ and hence $\mathcal{A} \in \mathcal{I}$. Therefore, there a one to one mapping between the set of feasible schedules $\mathcal{I}$ and the set of cliques in the cooperation graph.

Finally, let $\kappa$ be a clique. The weight can be expressed as follows:
\begin{align}
w(\kappa) &= \sum_{v_i \in \kappa} w(v_i) \nonumber \\
&= \sum_{v_i \in \kappa } \left( |\mathcal{C}_{i} \cap M_w| - |\{i\}\cap M_w| + y_{i}(\kappa_{i}^*) \right) .
\label{eqth32}
\end{align}
The weight in \eref{eqth32} is the objective function illustrated in \eref{eqth31}. As a conclusion, the maximum weight clique in the cooperation graph is the optimal solution to the optimization problem \eref{eq3} which conclude our proof.

\section{Proof of \lref{lem1}}\label{ap4}

To proof the uniqueness of the decomposition, this section first introduces an algorithm to perform the decomposition. Afterward, the rest of the section shows that the decomposition is unique.

\begin{algorithm}[t]
\begin{algorithmic}
\REQUIRE $\mathcal{A}, \mathcal{C}_{i}, \forall\ i \ in \mathcal{A}$.
\STATE Initialize $\mathcal{Z} = \varnothing$.
\STATE Initialize $k = 0$.
\WHILE{$\mathcal{A} \neq \varnothing$ }
\STATE Set $k \leftarrow k + 1$.
\STATE Initialize $Z_k = \varnothing$.
\STATE Set $Z_k \leftarrow \{ a \},\ a \in \mathcal{A}$.
\STATE Set $\mathcal{A} \leftarrow \mathcal{A} \setminus \{ a \}$.
\STATE Initialize $t = \TRUE$.
\WHILE{$t = \TRUE$}
\STATE Set $t \leftarrow \FALSE$.
\FORALL{ $b \in \mathcal{A}$}
\IF {$\mathcal{C}_b \cap \mathcal{C}^T(Z_k) \neq \varnothing$}
\STATE Set $Z_k \leftarrow Z_k \cup \{b\}$.
\STATE Set $\mathcal{A} \leftarrow \mathcal{A} \setminus \{ b \}$.
\STATE Set $t \leftarrow \TRUE$.
\ENDIF
\ENDFOR
\ENDWHILE
\ENDWHILE
\STATE Set $\mathcal{Z} \leftarrow \{Z_1,\ \cdots,\ Z_{k} \}$
\end{algorithmic}
\caption{Cluster Construction}
\label{algo1}
\end{algorithm}

To generate a clustering $\mathcal{Z}$ associated with the set of transmitting devices $\mathcal{A}$ respecting the conditions \eref{eq14}, \eref{eq5} and \eref{eq6} simultaneously, this section proposes a sequential constructing method. A cluster $Z_1$ is first generated by including an arbitrary element of the set $\mathcal{A}$. Afterward, all the remaining devices that are interfering with the cluster $Z_1$ are added to it. It is easy to see that such method of construction ensures that constraint \eref{eq6} holds. After this step all the remaining devices in $\mathcal{A}$ are not interfering with $Z_1$. A second cluster $Z_2$ is generated by including one of the transmitting devices in $\mathcal{A}$ and all the interfering devices added to it. Note that all the devices in $Z_2$ are not interfering with devices in $Z_1$. Therefore, $Z_1$ and $Z_2$ are not interfering and by extension the clustering $\mathcal{Z}$ satisfy the constraint \eref{eq5}. The process is repeated until $\mathcal{A} = \varnothing$ which ensures that constraint \eref{eq14} is satisfied. Therefore, the algorithm proves the existence of a clustering $\mathcal{Z}$ respecting the conditions \eref{eq14}, \eref{eq5} and \eref{eq6} simultaneously for all combination of transmitting devices $\mathcal{A}$. The steps of the algorithm are summarized in \algref{algo1}.

Now we proof uniqueness of such decomposition. Let $\mathcal{A}$ be a set of transmitting devices and let $\mathcal{Z}$ and $\mathcal{Z}^{\prime}$ be two distinct clustering of the set $\mathcal{A}$ satisfying the constraints \eref{eq14}, \eref{eq5} and \eref{eq6} simultaneously. From the constraint \eref{eq14}, we have:
\begin{align}
\bigoplus_{Z \in \mathcal{Z}} Z = A = \bigoplus_{Z^{\prime} \in \mathcal{Z^{\prime}}} Z^{\prime}.
\end{align}
Note that this condition implies only that the clustering have the same number of devices and not necessarily the same clusters $Z$. Since $\mathcal{Z} \neq \mathcal{Z}^{\prime}$, given that both clustering satisfy \eref{eq14} then $\exists\ i \in Z \cap Z^{\prime}$ with $Z \in \mathcal{Z},\ Z^{\prime} \in \mathcal{Z}^{\prime},\ Z \neq Z^{\prime}$.

Let $K$ be the set constructed by first including the device $i$. Each device $k$ in the set $\mathcal{A}$ is included in $K$ if $\mathcal{C}_{k} \cap \mathcal{C}^T(K) \neq \varnothing$. When no more devices can be added the process stop. Constraint \eref{eq5} can be written as:
\begin{align}
\mathcal{C}^T(Z) \cap \mathcal{C}^T(\mathcal{A} \setminus Z ) = \varnothing,\ \forall \ Z \in \mathcal{A}.
\end{align}
From \eref{eq5}, we can clearly see that $K \subseteq Z$ and $K \subseteq Z^{\prime}$. Assume that $K \subset Z$, then from \eref{eq6}, we have:
\begin{align}
K \subset Z \Leftrightarrow \mathcal{C}^T(Z \setminus K) \cap \mathcal{C}^T(K) \neq \varnothing.
\end{align}
From the construction of the set $K$, we have:
\begin{align}
\mathcal{C}^T(A \setminus K) \cap \mathcal{C}^T(K) = \varnothing.
\end{align}
Since $\mathcal{C}^T(Z \setminus K) \subset \mathcal{C}^T(A \setminus K)$, then $K=Z$. The same reasoning can be done to $Z^\prime$ and thus $Z = Z^\prime$ and $\mathcal{Z} = \mathcal{Z}^{\prime}$. This proves uniqueness of the decomposition and concludes our proof.

\section{Proof of \thref{th4}}\label{ap5}

To prove this theorem, we first develop each term of the optimization problem \eref{reformulation} in order to reformulate the problem in a more tractable form. Exploiting the constraint in the creation of the clusters $Z \in \mathcal{Z}$, we show that the problem can be formulated as an optimization over the individual clusters $Z$ that can be seen as ``virtual" devices in the network. Finally, using the results of \thref{th3}, we conclude that the optimal solution to the optimization problem \eref{reformulation} is the maximum weight clique in the full cooperation graph.

Using an analysis similar to the one done in \eref{eq8}, the set of devices out of the transmission range of the transmitting devices can be written as:
\begin{align}
\mathcal{S} \cap M_w &= M_w \setminus \left( \mathcal{C}^T(\mathcal{Z}) \cap M_w \right) .
\end{align}

Since $\mathcal{C}^T (Z) \cap \mathcal{C}^T (Z^\prime) = \varnothing,\ \forall\ Z \neq Z^\prime \in \mathcal{Z}$ (constraint \eref{eq5} on the generation of the clusters), then we can write $\mathcal{C}^T(\mathcal{Z}) = \bigoplus\limits_{Z \in \mathcal{Z}} \mathcal{C}^T(Z)$. Therefore, the cardinality of the set can be written as:
\begin{align}
|\mathcal{S} \cap M_w| = |M_w| - \sum_{Z \in \mathcal{Z}} | \mathcal{C}^T(Z) \cap M_w |.
\label{eq7}
\end{align}
As for \eref{eq2}, the first term in \eref{eq7} is constant. Therefore, it is removed from the optimization problem. Clearly, we have $|\mathcal{Z} \cap M_w| = \sum\limits_{Z \in \mathcal{Z}} | Z \cap M_w |$. Similarly, the set of devices in interference can be reformulated using the clusters $Z$ as follows:
\begin{align}
\mathcal{T}(\mathcal{Z}) = \bigcup_{Z,Z^\prime \in \mathcal{Z}} \bigcup_{\substack{z \in Z\\ \ z \neq z^\prime \in Z^\prime}} (\mathcal{C}_{z} \cap \mathcal{C}_{z^\prime}) .
\label{eqth41}
\end{align}
From the constraint \eref{eq5}, inter-cluster interference region is empty. In other words, we have:
\begin{align}
\bigcup_{Z \neq Z^\prime \in \mathcal{Z}} \bigcup_{\substack{z \in Z\\ \ z^\prime \in Z^\prime}} (\mathcal{C}_{z} \cap \mathcal{C}_{z^\prime}) = \varnothing .
\label{eqth42}
\end{align}
Substituting \eref{eqth42} in \eref{eqth41} yields the following expression of the interference region:
\begin{align}
\mathcal{T}(\mathcal{Z}) &=\bigcup_{Z \in \mathcal{Z}} \bigcup_{z \neq z^\prime \in Z} (\mathcal{C}_{z} \cap \mathcal{C}_{z^\prime}) \nonumber \\
&= \bigoplus_{Z \in \mathcal{Z}} \mathcal{T}(Z).
\end{align}
Therefore, the cardinality of the set can be expressed as follows:
\begin{align}
|\mathcal{T} \cap M_w| = \sum_{Z \in \mathcal{Z}} | \mathcal{T}(Z) \cap M_w |.
\end{align}
Finally, the optimization problem \eref{reformulation} can be formulated as follows:
\begin{align}
\mathcal{Z}^* &= \arg \max_{ \mathcal{Z} \in \mathfrak{Z} } \sum_{Z \in \mathcal{Z}} |\mathcal{C}^T(Z) \cap M_w| - |Z \cap M_w| \nonumber \\
& \qquad - |\mathcal{T}(Z) \cap M_w| + y(Z)
\end{align}
where the function $y(Z)$ is defined as follows:
\begin{align}
y(Z) = \sum_{i \in Z} \left[ \max_{\kappa_i \in \mathcal{G}_i} \left(\sum_{j \in \tau_i(\kappa_i)} (1 - p_{ij} ) \right) \right]
\end{align}

Since the clusters are non-interfering clusters, they can be seen as new \emph{devices} in the network. Using a methodology very close to the one used in \thref{th3}, the optimal clustering $\mathcal{Z}^*$ is given by the maximal weight clique in the full cooperation graph where the weight of each vertex $v$ representing the cluster $Z$ can be expressed as:
\begin{align}
v =& |\mathcal{C}^T(Z) \cap M_w| - | Z \cap M_w| - |\mathcal{T}(Z) \cap M_w| + y(Z). \nonumber
\end{align}

\section{Proof of \thref{th5}}\label{ap6}

To prove the theorem, it is sufficient to show that all the vertices in the maximum weight clique in the cooperation graph generated in \sref{cooper} are created in the cooperation graph produced by \algref{algo2}. This is because the cooperation graph produced by \algref{algo2} is a sub-graph of the full cooperation graph generated in \sref{cooper}. Let $\mathcal{Z}^*=\{Z^*_1,\ \cdots,\ Z^*_{|\mathcal{Z}^*|}\}$ be the maximum weight clique in the cooperation graph generated in \sref{cooper}. Let $Z^*_k=\{z^k_1,\ \cdots,\ z^k_{|Z^*_k|}\} \in \mathcal{Z}^*$ be a cluster where the element has been arranged such that:
\begin{align}
\mathcal{C}_{z^k_i} \cap \mathcal{C}^T(\{z^k_1,\ \cdots,\ z^k_{i-1}) \neq \varnothing,\ \forall\ 1 \leq i \leq |Z^*_k|.
\end{align}

We prove by induction that all the clusters $v_i=\{z^k_1,\ \cdots,\ z^k_{i}\},\ \forall\ 1 \leq i \leq |Z^*_k|$ are generated by \algref{algo2}. For $i=1$, the cluster (device) $z^k_1$ is included in the graph by construction. Assume that cluster $v_{i-1}$ is included. From the analysis done in \eref{eq10}, device $z^k_i$ should provide a gain to be considered in the cluster. However, this result holds only for clusters that satisfy \eref{eq16}. We generalize the result for the clusters in the maximum weight clique. Note that the generalization is possible only for the maximum weight clique of the graph.

Assume that $w(v_i) \leq w(v_{i-1})$, then we have $w(Z^*_k) \leq w(Z^*_k \setminus \{z^k_i\})$ and the cluster $Z^*_k \setminus \{z^k_i\}$ is combinable with all the clusters $Z^*_i,\ \forall\ 1 \leq i \neq k \leq |\mathcal{Z}^*|$ and hence $\mathcal{Z}=\{Z^*_1,\ \cdots,\ Z^*_{k-1},\ Z^*_k \setminus \{z^k_i\},\ Z^*_{k+1}\ \cdots,\ Z^*_{|\mathcal{Z}^*|}\}$ will yields a higher weight. Since $\mathcal{Z}^*$ is the maximum weight clique then we conclude that $w(v_i) \geq w(v_{i-1})$. The same proof can be done if $w(v_i) \leq w(z^k_i)$ in which case
$\mathcal{Z}=\{Z^*_1,\ \cdots,\ Z^*_{k-1},\ Z^*_k \setminus \{v_{i-1}\},\ Z^*_{k+1}\ \cdots,\ Z^*_{|\mathcal{Z}^*|}\}$ will yields a higher weight. In other words, we have:
\begin{align}
w(v_i) \geq \max (w(v_{i-1}),w(z^k_i)).
\end{align}

Therefore, cluster $v_i$ is created. This concludes that all the clusters in the maximum weight clique are generated in the cooperation graph resulting from \algref{algo2}. All the connections are the same in both graphs. Therefore, the maximum weight clique in the cooperation graph created in \sref{cooper} is the same as the one in the graph generated by \algref{algo2}.

\section{Auxiliary Results} \label{app10}

\subsection{Partition of the Set of devices}

\begin{lemma}
For any set of transmitting devices $\mathcal{A}$, the following sets forms a partition of the set of all devices $\mathcal{M}$:
\begin{align}
&\mathcal{M} = \bigoplus_{i \in \mathcal{A}} \left( \mathcal{O}_{i} \cap M_w \right) \oplus \left( \mathcal{A} \cap M_w \right) \nonumber \\
& \qquad \qquad \oplus \left( \mathcal{T} \cap M_w \right) \oplus \left( \mathcal{S} \cap M_w \right) \oplus \overline{M}_w.
\end{align}
\label{l2}
\end{lemma}

\begin{proof}
We first prove that for any set of transmitting devices $\mathcal{A}$, the following results holds:
\begin{align}
\mathcal{M} = \bigcup_{i \in \mathcal{A}} \mathcal{O}_{i} \cup \mathcal{A} \cup \mathcal{T} \cup \mathcal{S}.
\end{align}

Let $ \mathcal{C}^T = \bigcup\limits_{i \in \mathcal{A}} \mathcal{C}_{i}$ be the total coverage zone of the transmitting devices. By definition of the set $\mathcal{S}$, we have:
\begin{align}
\mathcal{M} = \mathcal{C}^T \cup \mathcal{S}.
\end{align}
Hence the opportunity zone can be written as follows:
\begin{align}
\mathcal{O}_i = \mathcal{C}_i \setminus \left( \mathcal{A} \cup \mathcal{T} \right) , \forall \ i \in \mathcal{A}.
\end{align}
Developing the union of the opportunity zones yields:
\begin{align}
\bigcup_{i \in \mathcal{A}}\mathcal{O}_{i} &= \bigcup_{i \in \mathcal{A}}\mathcal{C}_{i} \setminus \left( \mathcal{A} \cup \mathcal{T} \right) \nonumber \\
\bigcup_{i \in \mathcal{A}} \mathcal{O}_{i} &= \mathcal{C}^T \setminus \left( \mathcal{A} \cup \mathcal{T} \right) \nonumber \\
\mathcal{C}^T &= \bigcup_{i \in \mathcal{A}} \mathcal{O}_{i} \cup \mathcal{A} \cup \mathcal{T}.
\end{align}
Therefore, we obtain the desired result:
\begin{align}
\mathcal{M} = \bigcup_{i \in \mathcal{A}} \mathcal{O}_{i} \cup \mathcal{A} \cup \mathcal{T} \cup \mathcal{S}.
\end{align}

We now prove that the sets are disjoint. By definition of the sets $\mathcal{A}$, $\mathcal{T}$ and $\mathcal{S}$, we have that these sets are pairwise disjoint. Also by definition, we have the set $\mathcal{O}_{i}$ is disjoint from $\mathcal{A}$ and $\mathcal{T},\ \forall i \in \mathcal{A}$. Finally, assume $\exists \ i \in \mathcal{O}_{j} \cap \mathcal{O}_{k}$ for $(j , k) \in \mathcal{A}^2$. We have $\mathcal{O}_{i} \subset \mathcal{C}_{i}$. Therefore, $i \in \mathcal{C}_{j} \cap \mathcal{C}_{k}$. By definition of the interference set $\mathcal{T}$, we have $i \in \mathcal{T}$. By definition of the opportunity zone, we have $\mathcal{O}_{i} \cap \mathcal{T} = \varnothing,\ \forall\ i \in \mathcal{A}$. Hence such $i \in \mathcal{O}_{j} \cap \mathcal{O}_{k}$ does not exist and all the sets $\mathcal{O}_{l}, \forall\ l \in \mathcal{A}$ are pairwise disjoint. Obviously, we have $ ( X \cap Y ) \cap \overline{Y} = \varnothing, \ \forall\ X,Y$. This conclude our proof and we have:
\begin{align}
&\mathcal{M} = \bigoplus_{i \in \mathcal{A}} \left( \mathcal{O}_{i} \cap M_w \right) \oplus \left( \mathcal{A} \cap M_w \right) \nonumber \\
& \qquad \qquad \oplus \left( \mathcal{T} \cap M_w \right) \oplus \left( \mathcal{S}\cap M_w \right) \oplus \overline{M}_w.
\end{align}
\end{proof}

\subsection{No Interference Opportunity Zone}

\begin{lemma}
For any set of transmitting device $\mathcal{A} \in \mathcal{I}$ satisfying the non-interference constraint, we have:
\begin{align}
\mathcal{O}_i(\mathcal{A}) = \mathcal{O}_i = \mathcal{C}_i \setminus \{i\}.
\end{align}
\label{l3}
\end{lemma}

\begin{proof}
Let $\mathcal{A} \in \mathcal{I}$ satisfying the non interference constraint, we have by definition of the interference set $\mathcal{T}(\mathcal{A}) = \varnothing$. Hence, the expression of the opportunity zone $\mathcal{O}_i(\mathcal{A})$ of any device $i \in \mathcal{A}$ becomes:
\begin{align}
\mathcal{O}_i(\mathcal{A}) &= \mathcal{C}_i \setminus (\mathcal{A} \cup \mathcal{T}(\mathcal{A})) \nonumber \\
&= \mathcal{C}_i \setminus \mathcal{A}.
\end{align}

We now prove that $(\mathcal{C}_i\setminus \{i\}) \cap \mathcal{A} = \varnothing$. Assume $\exists \ j \in (\mathcal{C}_i\setminus \{i\}) \cap \mathcal{A}$. Then, by definition of the coverage zone, we have $j \in \mathcal{C}_j$ and thus $j \in \mathcal{C}_j \cap (\mathcal{C}_i\setminus \{i\}) \in \mathcal{T}(\mathcal{A})$. However, by definition of the set $\mathcal{A}$, the interference region is empty $\mathcal{T}(\mathcal{A}) = \varnothing$. Therefore, such $j \in (\mathcal{C}_i\setminus \{i\}) \cap \mathcal{A}$ does not exist, which conclude our proof and we have:
\begin{align}
\mathcal{O}_i(\mathcal{A}) = \mathcal{O}_i = \mathcal{C}_i\setminus \{i\}.
\end{align}
\end{proof}

\end{document}